\documentclass[reqno]{amsart}
\usepackage{amssymb}
\usepackage[dvips]{epsfig}
\usepackage{graphicx}
\usepackage{color}
\usepackage{amsmath}
\usepackage{amssymb}
\usepackage{amsfonts}
\usepackage{amsthm}
\usepackage{bbm}
\usepackage{mathrsfs}
\newcommand{\R}{\mathbb R}

\newcommand{\Z}{\mathbb Z}

\newcommand{\C}{\mathbb C}

\newcommand{\ep}{\varepsilon}

\newcommand{\g}{\gamma}

\renewcommand{\l}{\lambda}

\newcommand{\N}{\mathbb{N}}
\newcommand{\T}{\mathbb{T}}

\newtheorem{thm}{Theorem}[section]
\newtheorem{lem}[thm]{Lemma}

\theoremstyle{remark}
\newtheorem{rem}{\bf Remark}[section]
\theoremstyle{definition}
\newtheorem{defn}[thm]{Definition}

\numberwithin{equation}{section}
\usepackage[colorlinks=true,pdfstartview=FitV,linkcolor=magenta,citecolor=cyan]{hyperref}

\usepackage{bm}
\begin{document}

\title[]{ Diagonalization  in  a  quantum kicked rotor model  with  non-analytic potential}
\author[Y. Shi]{Yunfeng Shi}
\address[YS] {School of Mathematics,
Sichuan University,
Chengdu 610064,
China}
\email{yunfengshi@scu.edu.cn}
\author[L.Wen]{Li Wen}
\address[LW]{School of Mathematics,
Sichuan University,
Chengdu 610064,
China}
\email{liwen.carol98@gmail.com}

\date{\today}


\begin{abstract}
In this paper we study the lattice quasi-periodic operators with power-law long-range hopping and meromorphic monotone potentials,  and diagonalize the operators via a Nash-Moser iteration scheme.  As applications, we  obtain uniform power-law localization,  uniform  dynamical localization and Lipschitz continuity of the integrated density of states (IDS) for such operators.  Our main motivation comes from investigating quantum suppression of chaos in  a  quantum kicked rotor model with  non-analytical potential.
\end{abstract}

\maketitle

\maketitle
\section{Introduction}
Quantum chaos  aims to  investigate the quantum mechanics of classically chaotic systems. One of  the basic models in quantum chaos is the so called quantum kicked rotor, which was first introduced by \cite{Cas79} as a quantum analog of the standard mapping.   This model is defined  by
\begin{align}\label{qkrm}
\sqrt{-1}\frac{\partial \Psi(\theta, t)}{\partial t}= \left(\mathcal{H}_0+\check{\phi}(\theta)\sum_{n\in\Z}\delta(t-2n\omega)\right)\Psi(\theta, t),
\end{align}
where
\begin{align*}
\ (\theta, t)\in\R\times \R, \ \omega>0, \ 
\mathcal{H}_0=\frac{\partial^2 }{\partial \theta^2}\ {\rm or} \ \sqrt{-1}\frac{\partial}{\partial \theta},
\end{align*}
and $\check{\phi}:\ \T=\R/\Z\to\R$ is a  potential\footnote{The original model in   \cite{Cas79} corresponds to \eqref{qkrm} with $\mathcal{H}_0=\frac{\partial^2 }{\partial \theta^2}$ and $\check{\phi}(\theta)=\cos(2\pi \theta)$.}. 
Typically, the motion in  the  quantum  kicked rotor model   is  almost-periodic while the classical one  is   chaotic. This quantum suppression of chaos phenomenon   was  first  discovered in \cite{Cas79}, and well understood after the 
 remarkable work of Fishman, Grempel and Prange (cf. \cite{FGP82, GFP82}),  in which they mapped the quantum kicked rotor model onto  lattice ergodic operators $H_x$  given by 
 \begin{align*}
(H_xu)(n)=\sum_{m\neq n}\phi({n-m})u(m)+d_n(x)u(n),\  x\in\R,\ n\in\Z, 
\end{align*}
where
\begin{align*}
\phi(n)&=-\int_{0}^1\tan\pi\left(\frac{\check{\phi}(\theta)}{2}\right)e^{-2\pi\sqrt{-1} n\theta}{d}\theta,\\
d_n(x)&=\tan\pi\left(x-n^2\omega\right)\ {\rm or}\ \tan\pi\left(x-n\omega\right)\ \mbox{depending on} \ \mathcal{H}_0.
\end{align*}
The parameter $x\in\R$ represents the quasi-energy of the Floquet operator associated to \eqref{qkrm} (cf. \cite{FGP82}).   It turns out   that the almost-periodicity of the  motion in \eqref{qkrm} follows from  the  Anderson localization (i.e., pure point spectrum with exponentially decaying eigenfunctions)  of  $H_x$  for a.e. $x\in\R$ (cf. \cite{FGP82,SPK22}). 

For the  analytic potential $\check{\phi}$, the operator $H_x$ admits an exponential  hopping (i.e.,  $|\phi(n)|\leq e^{-c|n|}$ for some  $c>0$). 
However,  the non-analytic (even singular) potential $\check{\phi}$ which could yield a power-law  decay  hopping  (i.e.,  $|\phi(n)|\leq {|n|^{-C}}$ for some  $C>0$) also appears naturally in some important physical models, e.g., the quantum Fermi accelerator (cf. \cite{JC86}). The  work \cite{GW05} has provided numerical evidence  for  the    localization transition  in quantum chaos with  the  singular potential  $\check{\phi}(\theta)=(|\theta|^\alpha\mod 1$), which induces a  hopping of  $|\phi(n)|\sim |n|^{-1-\alpha}$. Note  that the analytic derivation of the localization  in \cite{GW05} relies on physical perspective of localization  for random operators with power-law hopping (cf. \cite{Rod03}).  However, the  on-site energy sequence $\{d_n(x)\}_{n\in\Z}$ is only pseudo-random and the localization properties of $H_x$ should depend  on arithmetic properties of $\omega$. Thus, a mathematically rigorous  treatment of the localization for quasi-periodic operators with (realistic) power-law  hopping becomes important. We would also like to mention that certain quantum kicked rotor models  in higher dimensions (cf. \cite{DF88}) or with a quasi-periodic potential (cf. \cite{CGS89})  will give rise to higher dimensional lattice quasi-periodic operators with tangent potential.   

The main purpose  of  this paper is to diagonalize  certain lattice quasi-periodic operators with power-law decay hopping and meromorphic monotone potentials  via the  Nash-Moser  iteration scheme.  As applications, we  prove  uniform power-law localization,  uniform  dynamical localization and Lipschitz continuity of the IDS.   
More precisely, consider on $\Z^d$ the operators
 \begin{align}\label{model0}
H_z=\varepsilon T_\phi+f(z-\bm n\cdot \bm\omega)\delta_{\bm n\bm n'},\ \varepsilon\in\R, 
 \end{align}
 where the off-diagonal part (i.e., the hopping term) $T_\phi$  satisfies  
 \begin{align*}
(T_\phi u)({\bm n})=\sum_{\bm m\neq\bm n}\phi({\bm n-\bm m})u({\bm m}),\ \phi({\bm 0})=0,\ |\phi({\bm n})|\leq |\bm n|^{-s}
\end{align*}
 with some $s>0$ and $|\bm n|=\max\limits_{1\leq i\leq d}|n_i|$.   In the diagonal part we assume the potential $f$ is  a $1$-periodic meromorphic function   defined on ${D}_{R}=\{z\in\C:\ |\Im z|<R\}$  satisfying the monotonicity  condition (cf.  \eqref{gc} for details),  which includes $\tan (\pi z)$ and $e^{2\pi \sqrt{-1}z}$ as the special cases.  We call  $\bm\omega\in[0,1]^d$ the frequency  and $z$ the phase. Typically, we  assume $\bm \omega$ satisfies the Diophantine condition, namely, there exist $\tau>d$ and $\gamma>0$ so that 
 \begin{align*}
\|\bm n\cdot\bm \omega\|_\T=\inf_{l\in\Z}|l-\bm n\cdot\bm \omega|\geq \frac{\gamma}{|\bm n|^\tau}\ {\rm for}\ \forall\ \bm n\in\Z^d\setminus\{\bm 0\}.
 \end{align*}
 Obviously,   if $\varepsilon=0$, then $H_z$ exhibits the Anderson localization for all $z$  such that   $z-\bm n\cdot\bm \omega$ is not a pole of $f$ for all $\bm n\in\Z^d$.  So it is reasonable to expect the  localization of $H_z$  for $0<|\varepsilon| \ll1$.  A  potential approach is to diagonalize $H_z$ via some unitary transformation. However,  it is not totally trivial to carry out such a method  since  the small divisors problem, i.e., 
 \begin{align*}
\liminf_{|\bm n|\to \infty}|f(z-\bm n\cdot \bm\omega)-f(z)|=0. 
 \end{align*} 
 In 1983 Craig \cite{Cra83} first proved an inverse Anderson localization result for some lattice almost-periodic Schr\"odinger operators by using  a KAM type diagonalization method.   This   method of Craig  is to fix  a diagonal operator $d_{\bm n}\delta_{\bm n\bm n'}$ (with $d_{\bm n}=g(\bm n\cdot\bm \omega)$ for some 1-periodic function $g$) satisfying  
 \begin{align}\label{NR0}
 |d_{\bm m}-d_{\bm n}|>\frac{\gamma}{|\bm m-\bm n|^\tau}\ (\bm m\neq \bm n).
 \end{align}
  An additional condition is imposed on $g$ so that $\{d_{\bm n}\}_{\bm n\in\Z^d}$ is almost-periodic.   For sufficiently small $\varepsilon$,  Craig constructed a unitary transformation $Q$ and some 1-periodic function $g'$ so that 
 \begin{align*}
Q^{-1}(\varepsilon \Delta+g(\bm n\cdot\bm \omega)\delta_{\bm n\bm n'}+g'(\bm n\cdot\bm\omega)\delta_{\bm n\bm n'})Q=g(\bm n\cdot\bm \omega)\delta_{\bm n\bm n'},
 \end{align*}
 where $\Delta$ denotes the lattice Laplacian. The transformation $Q$ is  obtained   via  KAM iterations.   It is important that the modulation operator  $g'(\bm n\cdot\bm\omega)\delta_{\bm n\bm n'}$ is constructed so that  the small-divisors   $d_{\bm n}-d_{\bm m}$  do not change  in the iterations.  However,  it is much more difficult to deal with the direct problem, i.e., to diagonalize $\varepsilon \Delta+d_{\bm n}\delta_{\bm n\bm n'}$ via some unitary transformation.  In this case one has to deal with small divisors of the form 
 $\tilde d_{\bm n}-\tilde d_{\bm m}$ with $\tilde d_{\bm n}$ being the  modulation of $d_{\bm n}$.   Generally, it is hard to determine  whether  the modulated sequence  $\{\tilde d_{\bm n}\}_{\bm n\in\Z^d}$  still satisfies the  non-resonant condition \eqref{NR0} or not.     For this reason, Bellissard-Lima-Scoppola \cite{BLS83}  provided  a  class of  meromorphic monotone functions   which is stable under small analytic function  perturbations.  As a result, the non-resonant condition \eqref{NR0} can be essentially preserved  under  modulations.   Then by using the  KAM iteration method, \cite{BLS83} proved the Anderson localization for \eqref{model0} with $\phi$ being exponentially decaying. Later, P\"oschel \cite{Pos83} developed a general approach to study  both inverse and direct Anderson localization in the setting of translation invariant Banach algebra. He also presented new examples of limit-periodic Schr\"odinger operators that exhibit Anderson localization. We want to remark that the methods of \cite{Cra83,BLS83,Pos83} can only handle operators with exponential decay  hopping. Very recently, Shi \cite{Shi22} developed a Nash-Moser iteration type diagonalization approach to deal with almost-periodic  operators with power-law decay hopping.   The key new ingredients of the proof  in \cite{Shi22} consist of  introducing the smoothing operation in the iteration together with establishing  the tame estimates for some operator norms induced by the power-law hopping.  Consequently, Shi \cite{Shi22} generalized the results of P\"oschel \cite{Pos83} to the power-law hopping case.  Finally, we  also refer to \cite{ FP84, Sim85,JL17, JK19, Kac19, KPS22}  for important  progress on the study of  lattice quasi-periodic  operators with monotone type  potentials.

 We would like to remark  that while the method  of  Shi \cite{Shi22}  works  for more general translation invariant Banach algebras,  if one applies  it to the operators  \eqref{model0},  the power-law localization of $H_z$ only  holds   if $|\varepsilon| \leq \varepsilon_0$ for some postive $\varepsilon_0=\varepsilon_0(z)$ depending on $z$.    This single-phase type localization result is obviously not satisfactory and one  expects uniform (in $z$) localization results  as that obtained  in \cite{BLS83}. In addition,  if we want to investigate  the dynamical localization and the regularity of the IDS,  a uniform in $z$ estimate of $\varepsilon_0$ is required. In this paper we combine the Nash-Moser iteration diagonalization method of Shi \cite{Shi22} with the meromorphic monotone  function estimates  of Bellissard-Lima-Scoppola \cite{BLS83} to establish uniform  power-law localization,  uniform  dynamical localization and Lipschitz continuity of the IDS for \eqref{model0}.  One of the new ingredients of our proof  is to introduce a  norm which  takes account of both $z\in {D}_R$ and $\bm n\in\Z^d$ directions.  This  new  norm is actually  the  $\ell^1$-type  one  rather than the  $\ell^2$-one as that in \cite{Shi22}, which will  lead   to  a significant simplification  of the proof.

\subsection{Main results}
We present our results in the language of holomorphic kernels algebra introduced by Bellissard-Lima-Scoppola \cite{BLS83}.

We first introduce the function class of potentials. For $R>0$, $\mathscr{H}_{R}$ denotes the set of period $1$ holomorphic bounded functions on ${D}_{R}=\{z\in\mathbb{C}:\ \left|\Im z\right|<R\}$,
equipped with the norm $ \|f  \|_{R}=\underset{z\in{D}_{R}}{\sup}|f(z)|$.
Then $\mathscr{P}_{R}$ denotes the set of period $1$ meromorphic functions $g$ on ${D}_{R}$ such that, there is a constant $c>0$ satisfying
\begin{align}\label{gc}
\left|g(z)-g(z-a)\right|\ge c\|a\|_{\T}\  {\rm for}\ \forall \  a\in\R\  {\rm and}\ \forall\  z\in{D}_{R}
\end{align}
with  $\|a \|_{\T}=\underset{k\in\Z}{\inf}|a-k|.$
Then $|g|_{R}$ is defined as the biggest  possible value of $c$ in \eqref{gc}. Let $\mathscr{P}=\underset{R>0}{\bigcup}{\mathscr{P}}_{R}$. We remark that  both $g_{1}(z)=e^{2\pi \sqrt{-1}z}$ (cf. \cite{Sar82}) and $g_{2}(z)=\tan (\pi z)$ are in $\mathscr{P}$ (cf. \cite{BLS83} for details).

 We then consider the holomorphic kernels  constituting the hopping operators.   Let
\begin{align*}
\bm\omega=(\omega_{1},\cdots,\omega_{d})\in\mathbb{R}^{d},\ \bm n=(n_{1},\cdots,n_{d})\in\Z^{d}.
\end{align*}
Denote $\bm n\cdot\bm\omega=\sum\limits_{i=1}^{d}n_{i}\omega_{i}$  and $|\bm n|=\underset{1\le i \le d}{\max}|n_{i}|$. For $R>0$ and $s\ge0$, denote by  $\mathscr{U}_{R,s}$  the set of kernels $\mathcal{M}=(\mathcal{M}(z,\bm n))_{\bm n\in\Z^{d},z\in{D}_{R}}$ with  $\mathcal{M}(z, \bm n)\in\mathscr{H}_{R}$ for each $\bm n\in\Z^{d}$,  and
\begin{align}\label{norm}
\|\mathcal{M} \|_{R,s}=\underset{z\in{D}_{R}}{\sup}\sum_{\bm{n}\in\Z^{d}}|\mathcal{M}(z,\bm{n})|\langle\bm n\rangle^{s}<\infty,\ \langle\bm n\rangle=\max\{1,|\bm n|\}.
\end{align}
 Then $\mathscr{U}_{R,s}$ is a Banach space. Giving $\bm\omega\in\mathbb{R}^{d}$, define  an algebraic structure with respect to (w.r.t) $\bm\omega$ by\footnote{This definition is also valid for elements that are not in $\mathscr{U}_{R,s}$, e.g.,     $\mathcal{M}=(\mathcal{M}(z,\bm n))_{\bm n\in\Z^{d},z\in{D}_{R}}$ with  $\mathcal{M}(z, \bm n)\in\mathscr{P}_{R}$ for each $\bm n\in\Z^{d}$. }
\begin{align}\label{prodm}
(\mathcal{M}_{1}\mathcal{M}_{2})(z,\bm n)=\sum_{\bm l\in\Z^{d}}\mathcal{M}_{1}(z,\bm l)\mathcal{M}_{2}(z-\bm l\cdot\bm\omega, \bm n-\bm l).
\end{align}
An involution w.r.t $\bm\omega$ is given by
\begin{align}\label{*}
(\mathcal{M}^{*})(z,\bm n)=\overline{\mathcal{M}(\bar{z}-\bm n \cdot\bm\omega,-\bm n)}.
\end{align}
By \eqref{norm} and \eqref{*}, we have $(\mathcal{M}_{1}\mathcal{M}_{2})^{*}=\mathcal{M}_{2}^{*}\mathcal{M}_{1}^{*}$ and $ \|\mathcal{M} \|_{R,s}=\|\mathcal{M}^{*} \|_{R,s}$. If $f$ is defined in some domain of $\mathbb{C}$, we also write $f^{*}(z)=\bar{f}(\bar{z})$. Then $\mathscr{U}_{R,s}$ with above two structures w.r.t  $\bm\omega$ denotes $\mathscr{U}_{R,s}^{\bm\omega}$.  
\begin{rem}
We remark that
\begin{itemize}
\item[(i).] If $f: D_R\to \C$, then $f$ can be considered as a diagonal kernel by defining
${f}(z,\bm n)\equiv f(z)\delta_{\bm n\bm 0}$.  So if $V\in \mathscr{P}$, we can still  define the product $V \mathcal{M}$  via \eqref{prodm} for  $\mathcal{M}\in \mathscr{U}_{R,s}^{\bm\omega}$.  We say $V\in\mathscr{P}$ is self-adjoint if $V^{*}=V.$
\item[(ii).] If $\bm e\in\Z^{d}$, $\mathcal{U}_{\bm e}$ is the kernel $\mathcal{U}_{\bm e}(z,\bm n)=\delta_{\bm n\bm e},$
and the Laplace kernel is then given by $\bm\delta=\sum\limits_{|\bm e|_1=1}\mathcal{U}_{\bm e},$ where $|\bm e|_1=\sum\limits_{i=1}^d|e_i|$.
\item[(iii).] The unit in  $\mathscr{U}_{R,s}^{\bm\omega}$ is denoted by $\bm 1=\mathcal{U}_{\bm 0}$. Then $\mathcal{M}\in\mathscr{U}_{R,s}^{\bm\omega}$ is called unitary (resp. self-adjoint) if $\mathcal{M}\mathcal{M}^*=\mathcal{M}^*\mathcal{M}=\bm1$ (resp. $\mathcal{M}=\mathcal{M}^*$).

\end{itemize}
\end{rem}

Finally, we define the Diophantine condition.  Fix $\gamma>0$ and $\tau>d$. We say $\bm\omega\in[0,1]^d$ satisfies  the $(\tau, \gamma)$-Diophantine condition,  if 
	\begin{align*}
		\|\bm n\cdot\bm\omega  \|_{\T} \ge \frac{\g}{\langle\bm n\rangle^{\tau}}\ {\rm for}\ \forall\ \bm n\in\Z^d\setminus\{\bm 0\}.
	\end{align*}
Denote by ${\rm DC}_{\tau,\gamma}$  the set of all  $(\tau,\gamma)$-{Diophantine}  frequencies. 

Throughout this paper  we always assume $\bm\omega\in{\rm DC}_{\tau,\g}.$  {This assumption is reasonable  since the long-range hopping in our model has  a power-law  off-diagonal decay.  More precisely, since we use a KAM type perturbation method,  imposing  Diophantine condition on the frequency  allows  a controllable loss of derivatives (cf. Lemma \ref{hollem} in this paper) when solving the homological equations.  Without this assumption, we can not have effective estimates on the solutions of homological equations, and the iteration scheme will become  invalid.  In contrast, in the work of \cite{Cra83,Pos83} the Diophantine frequency can be relaxed to the Bruno-R\"ussmann one \footnote{For example, this  condition  includes  the case that the power-law bound in the Diophantine condition is replaced by a sub-exponential one. } since  the hopping operators in  \cite{Cra83, Pos83}  are of exponential off-diagonal decay.   However, if the  power-law lower bound in the Diophantine condition is replaced by an exponential one (i.e., the  Liouville frequency case),  to the best of our knowledge,  there is  simply no localization results  for quasi-periodic operators on $\Z^d$.  }


 We are able to state our main results.
\subsubsection{Diagonalization}
We first introduce  the diagonalization theorem. 
\begin{thm}\label{mthm}
 Let $\bm\omega\in {\rm DC}_{\tau,\gamma}$. Fix $\delta>0$, $\alpha_{0}>0$ and $R>0$. Let 
\begin{align*}
	V\in \mathscr{P}_{R},\ \mathcal{M}\in\mathscr{U}_{R,\alpha+3\delta}^{\bm\omega}.
\end{align*}
Then for
\begin{align*}
	\alpha>\alpha_{0}+\tau+4\delta, 
\end{align*}
there is some $\ep_{0}=\ep_{0}(R,\alpha,\alpha_{0},\g,\tau,|V|_{R},\delta)>0$ such that the following holds true. If $\|\mathcal{M}\|_{R,\alpha+3\delta}<\ep_{0}$, then there exist an invertible element $\mathcal{U}\in\mathscr{U}_{\frac{R}{2}, \alpha-\tau-4\delta}^{\bm\omega}$ and some $\hat{V}\in\mathscr{P}_{\frac{R}{2}}$ so that
\begin{align}
\label{uvu}&\mathcal{U}(V+\mathcal{M})\mathcal{U}^{-1}=\hat{V},\\
\label{K1}& \|\mathcal{U}^{\pm1}-\bm 1 \|_{\frac{R}{2},\alpha-\tau-4\delta}\le K_{1}\|\mathcal{M}\|_{R,\alpha+3\delta}^{\frac{3\delta}{\alpha-\alpha_{0}}},\\
\label{vr2}&V-\hat{V}\in\mathscr{H}_{\frac{R}{2}}, \  |\hat{V}|_{\frac{R}{2}}\ge \frac{1}{2}|V|_{R},\\
\nonumber&\|V-\hat{V}  \|_{\frac{R}{2},0}\le K_{2}\|\mathcal{M}\|_{R,\alpha+3\delta},
\end{align}
where $K_1>0$ and $K_2>0$ only depend on $\alpha_{0},\alpha,\delta$.
Moreover,  if  both $\mathcal{M}$ and $V$ are self-adjoint, then $\mathcal{U}$ is unitary and $\hat{V}^{*}=\hat{V}$.
\end{thm}
\begin{rem}
In this theorem it suffices to assume $\mathcal{M}\in \mathscr{U}_{R, \alpha}^{\bm\omega}$ for $\alpha>\tau$ since $\delta>0$ and $\alpha_0>0$ can be  arbitrarily small.  
\end{rem}

\subsubsection{Power-law Localization}
Now we can apply the above theorem to obtain the (uniform)  {power-law} localization.

A  representation of $\mathcal{M}\in\mathscr{U}_{R,s}^{\bm\omega}$ in $\ell^{2}(\Z^{d})$ is given by
\begin{align}\label{tmz}
	\left(T_{\mathcal{M}}(z)\psi\right)(\bm n)=\sum_{\bm l\in\Z^{d}}\mathcal{M}(z-\bm n\cdot\bm \omega,\bm l-\bm n)\psi(\bm l),
\end{align}
where $\psi\in\ell^{2}(\Z^{d})$ and $z\in{D}_{R}$.


Fix $V\in \mathscr{P}_{R}$. Define for $0\le R'\leq R$ the set 
\begin{align}\label{ZR}
\mathcal{Z}_{R'}=\bigcap_{\bm n\in\Z^d}\{z\in \C:\ |\Im z| \leq R'\  {\rm and}\  z-\bm n\cdot\bm\omega {\rm \ is\  not\  a \ pole\  of}\  V \}.
\end{align}
Since $V$ is meromorphic, the set $D_R\setminus\mathcal{Z}_{R}$ is at most countable. For $V\in\mathscr{P}_{R}$ and $z\in\mathcal{Z}_R$,  denote by  $T_{V}(z)=V(z-\bm n\cdot\bm \omega)\delta_{\bm n\bm n'}$ the multiplication operator. 

 Then we have 
 \begin{thm}\label{pl}
	Let $\bm\omega\in {\rm DC}_{\tau,\gamma}$. Fix $\delta>0$, $\alpha_{0}>0$ and $R>0$. Let $V\in \mathscr{P}_{R}$ and $\mathcal{M}\in\mathscr{U}_{R,s}^{\bm\omega}$ with
	\begin{align*}
		s>\alpha_{0}+\tau+\frac{d}{2}+7\delta.
	\end{align*}
Then there is some $\ep_{0}=\ep_{0}(R,\alpha_{0},\g,\tau,|V|_{R},\delta,s,d)>0$ such that the following holds true.  If $\|\mathcal{M} \|_{R,s}<\ep_{0}$, then the operator $H_{z}=T_{\mathcal{M}}(z)+T_{V}(z)$ has a complete set of eigenfunctions $\{\varphi_{\bm n}\}_{\bm n\in\Z^d}$ obeying $|\varphi_{\bm n}(\bm i)|\le2\langle\bm n-\bm i\rangle^{-s+\tau+7\delta}$ for all $\bm n\in\Z^d$, $\bm i\in\Z^d$ and $z\in\mathcal{Z}_{{R}/{2}}$.
In addition, if  both $\mathcal{M}$ and $V$ are self-adjoint,  then $H_{x}$ is self-adjoint and its spectrum is equal to $\mathbb{R}$ for $x\in\mathcal{Z}_0$.
\end{thm}
\begin{rem}
We first mention that the perturbation strength $\varepsilon_0$ is independent of $x$ for $x\in\mathcal{Z}_0.$ 
\end{rem}
\begin{rem}
We  have explicit  descriptions of  localization centers. Moreover, we establish in fact the uniform (power-law) localization (cf. \cite{Jit97} for the definition of uniform localization).  
\end{rem}

\begin{rem}
Let $H_x=\varepsilon T_\phi+\tan\pi(x-\bm n\cdot\bm\omega)\delta_{\bm n\bm n'}$ for some sequence $\phi=\{\phi(\bm n)\}_{\bm n\in\Z^d}$
satisfying $\phi(\bm 0)=0$,  and $|\phi(\bm n)|\le |\bm n|^{-s}$ for $\bm n \neq \bm 0$.   If  $\bm\omega\in{\rm DC}_{\tau, \gamma},\ s>d+\tau$ and $|\varepsilon|\leq \varepsilon_0(s,\tau,\gamma,d)>0$,     then $H_x$ has uniform power-law localization for all  $x\not\in \frac12+\Z+\bm\omega\cdot\Z^d.$ This extends  the perturbative results of \cite{BLS83}   to  the power-law hopping case. 
\end{rem} 

\subsubsection{Uniform  dynamical localization} 

In this section we apply our diagonalization theorem   to  study  dynamical  localization.
For $\psi\in \C^{\Z^d}$ and  $s\ge0$,  define
	\begin{align*}
		\|\psi\|_{s}^{2}&=\sum_{\bm n\in\Z^d}|\psi(\bm n)|^2\langle\bm n\rangle^{2s}.
	\end{align*}
	Let  $\ell^{2}_s(\Z^d)$ denote  the set  of all  $\psi=\{\psi(\bm n)\}_{\bm n\in\Z^{d}}$ satisfying $\|\psi\|_{s}<\infty$. 

 Given  the  family  $(H_x)_{x\in\T}$ of self-adjoint operators defined on $\ell^2(\Z^d)$,  we are interested in  the estimate  of 
	\begin{align*}
	\|e^{-\sqrt{-1}tH_x}\psi\|_{q}\ {\rm for}\  \psi\in \ell^{2}_{q}(\Z^d). 	
	\end{align*}
We have
\begin{thm}\label{dl}
	Let $\bm\omega\in {\rm DC}_{\tau,\gamma}$. Fix $\delta>0$, $\alpha_{0}>0$, $R>0$ and $q\ge0$. Let both $V\in \mathscr{P}_{R}$ and $\mathcal{M}\in\mathscr{U}_{R,s}^{\bm\omega}$ be self-adjoint. Assume  
	\begin{align*}
		s>\alpha_{0}+\tau+q+\frac{d}{2}+7\delta.
	\end{align*}
	Then there is some $\ep_{0}=\ep_{0}(R,\alpha_{0},\g,\tau,|V|_{R},\delta,s,d,q)>0$ such that the following holds true.  If $\|\mathcal{M} \|_{R,s}<\ep_{0}$, then for  $\forall \ \psi\in \ell^{2}_{q}(\Z^d)$, we have 
\begin{align*}
\sup_{x\in \mathcal{Z}_0} \sup_{t\in\R}\|e^{-\sqrt{-1}tH_x}\psi\|_{q}<\infty,
\end{align*}
where   $H_{x}=T_{\mathcal{M}}(x)+T_{V}(x)$. 
\end{thm}

\begin{rem}
Since $\R\setminus\mathcal{Z}_0$ is at most countable, we have for $\forall \ \psi\in \ell^{2}_{q}(\Z^d)$, 
\begin{align*}
\int_{\T}\sup_{t\in\R}\|e^{-\sqrt{-1}tH_x}\psi\|_{q}dx<\infty,
\end{align*}
which refers to  the strong dynamical localization. 
\end{rem}

\subsubsection{Lipschitz continuity of the IDS} 

In this section prove the Lipschitz continuity of the IDS. 

	
	Let $V\in \mathscr{P}_{R}$ and $\mathcal{M}\in\mathscr{U}_{R,s}^{\bm\omega}$ with $V^{*}=V$, $\mathcal{M}^{*}=\mathcal{M}$. Let  $H_{x}=T_{\mathcal{M}}(x)+T_{V}(x)$ for $x\in\mathcal{Z}_0.$   Denote by $\mathbb{P}_{(-\infty,E]}(H_x)$ the spectral resolution of $H_x$, and $\chi_{L}$  the projection
\begin{align*}
	(\chi_{L}\psi)(\bm n)=\left\{
		\begin{array}{cc}
			\psi(\bm n) & \text{if}\ |\bm n|\le L,\\
			0 & \text{otherwsie},
		\end{array} 
	\right.
\end{align*}
respectively.  If $\bm\omega\in{\rm DC}_{\tau, \g}$, then the limit 
	\begin{align*}
		\kappa(E)=\lim\limits_{L\rightarrow\infty}\frac{1}{(2L+1)^{d}}\text{tr}(\chi_{L}\mathbb{P}_{(-\infty,E]}(H_x)) 
	\end{align*}
exists  and is independent of  $x$ for a.e. $x\in\T$.  We have 

\begin{thm}\label{idsc}
		Let  $\bm\omega\in {\rm DC}_{\tau,\g}$.  Fix $\delta>0$, $\alpha_{0}>0$ and $R>0$. 
	Assume further that \begin{align*}
		s>\alpha_{0}+\tau+d+7\delta.
	\end{align*}
	Then there is some $\ep_{0}=\ep_{0}(R,\alpha_{0},\g,\tau,|V|_{R},\delta,s,d)>0$ such that for $\|\mathcal{M} \|_{R,s}<\ep_{0}$ and $E_1, E_2\in\R$,  
	 $$|\kappa(E_1)-\kappa(E_2)|\leq \frac{2}{|V|_R}|E_1-E_2|.$$
\end{thm}
\begin{rem}
We refer to \cite{JK19, Kac19} for  \textit{all couplings} results on the Lipschitz continuity of IDS for  $1D$ lattice quasi-periodic Schr\"odinger operators with Lipschitz monotone potentials. 
\end{rem}
\subsection{Structure of the paper}
The paper is organized as follows. Some preliminaries including tame estimate and the smoothing operator are introduced in \S 2. The Nash-Moser iteration theorem is established in \S3. In \S 4 we prove the convergence of the iteration scheme, and then finish the proof of Theorem \ref{mthm}. The proofs of Theorems \ref{pl}, \ref{dl} and \ref{idsc} are completed in \S 5, \S 6 and \S 7, respectively. Some technical estimates are included in the appendix. 

\section{Preliminaries}

\subsection{Tame property} The norm defined by \eqref{norm} has the following important tame property.
\begin{lem}\label{ts}
	For any $s\ge0$ and $\mathcal{M}_{1},\mathcal{M}_{2}\in\mathscr{U}_{R,s}^{\bm\omega}$, we have
	\begin{align}\label{tame}
		\|\mathcal{M}_{1}\mathcal{M}_{2}\|_{R,s}\le K(s)(\|\mathcal{M}_{1}\|_{R,0}\|\mathcal{M}_{2}\|_{R,s}+\|\mathcal{M}_{1}\|_{R,s}\|\mathcal{M}_{2}\|_{R,0}),
	\end{align}
	where $K(s)=2^{\max(0,s-1)}$.  In particular,
	\begin{align}\label{s}
		\|\mathcal{M}_{1}\mathcal{M}_{2}\|_{R,0}\le\|\mathcal{M}_{1}\|_{R,0}\|\mathcal{M}_{2}\|_{R,0}.
	\end{align}
\end{lem}

	
\begin{proof}
	For a detailed proof, we refer to the appendix.
\end{proof}

\subsection{Smoothing operator} 
The smoothing operator plays an essential role in the Nash-Moser iteration scheme. In the present context we have
\begin{defn}\label{sm}
	Fix the $\theta\ge0$. Define the smoothing operator $S_{\theta}$ by
	\begin{align*}
		(S_{\theta}\mathcal{M})(z,\bm n)&=\mathcal{M}(z,\bm n)\ \text{for}\ |\bm n|\le\theta,\\
		(S_{\theta}\mathcal{M})(z,\bm n)&=0\ \text{for}\ |\bm n|>\theta.
	\end{align*}
\end{defn}
Given  a sequence $\{\theta_{l}\}_{l=0}^{\infty}$ with $\theta_{l+1}>\theta_{l}\ge0$ and $\lim\limits_{l\rightarrow\infty}\theta_{l}=+\infty$,  define
	\begin{align*}
		&\mathcal{M}^{(0)}=S_{\theta_{0}}\mathcal{M},\ \mathcal{M}^{(l)}=(S_{\theta_{l}}-S_{\theta_{l-1}})\mathcal{M} \ \text{for}\ l\ge1.
	\end{align*}
	Then $\mathcal{M}^{(l)}$ is called the  $l$-section of $\mathcal{M}$  w.r.t  $\{\theta_{l}\}_{l=0}^{\infty}$.
\begin{lem}
	Fix  the $\theta\ge0$. Then for $\mathcal{M}\in\mathscr{U}_{R,s}^{\bm\omega}$, we have
	\begin{align}
		\label{sm1}\|S_{\theta}\mathcal{M}\|_{R,s}&\le \langle\theta\rangle^{s-s'}\|\mathcal{M}\|_{R,s'}\  {\rm for}\ 0\le s'\le s,\\
		\label{sm2}\|(I-S_{\theta})\mathcal{M}\|_{R,s}&\le \langle\theta\rangle^{s-s'}\|\mathcal{M}\|_{R,s'}\  {\rm for}\ 0\le s\le s',
	\end{align}
	where $I$ denotes the identity operator.  In particular, if $\mathcal{M}^{(l)}$ is the  $l$-section of $\mathcal{M}$ w.r.t  $\{\theta_{l}\}_{l=0}^{\infty}$, we  have
	\begin{align}
		\label{m1l}\|\mathcal{M}^{(l)}\|_{R,s}&\le \langle\theta_{l}\rangle^{s-s'}\|\mathcal{M}\|_{R,s'}\  {\rm for}\ 0\le s'\le s,\\
		\label{m2l}\|\mathcal{M}^{(l)}\|_{R,s}&\le \langle\theta_{l-1}\rangle^{s-s'}\|\mathcal{M}\|_{R,s'}\  {\rm for}\ 0\le s\le s'.
	\end{align}
\end{lem}
\begin{proof}
	The proof follows immediately from  Definition \ref{sm} and \eqref{norm}.
\end{proof}

\section{The Nash-Moser iteration}
In this section we will prove a Nash-Moser iteration theorem. The main strategy is based on  the iteration scheme  outlined  in \cite{Shi22} combined with meromorphic function estimates of \cite{BLS83}.  The final transformation $\mathcal{U}$ will be obtained as the limit of  the product $\mathcal{U}_{l}=\prod\limits_{i=l}^{0}e^{\mathcal{W}_i}$ with a sequence of transformations ${\mathcal{W}_i}$ ($0\leq i\leq l$). More precisely, at the $l$-th iteration step we will find $\mathcal{W}_l\in \mathscr{U}_{R_{l},s}^{\bm\omega}$, $V_l\in \mathscr{P}_{R_l}$ and $\mathcal{R}_{l}\in \mathscr{U}_{R_{l},s}^{\bm\omega}$ so that
\begin{align*}
\prod_{i=l}^{0}e^{\mathcal{W}_{i}}\left(V+\sum_{i=0}^{l-1}\mathcal{M}^{(i)}\right)\prod_{i=0}^{l}e^{-\mathcal{W}_{i}}=V_{l}+\mathcal{R}_{l}
\end{align*}
and  $\|\mathcal{R}_{l}\|_{R_{l},s}=o(\|\mathcal{R}_{l-1}\|_{R_{l-1},s})$, where  $\mathcal{M}^{(l)}$ is  the $l$-section of $\mathcal{M}$ w.r.t  $\{\theta_{l}=\theta_{0}\Theta^{l}\}_{l=0}^{\infty}$ with $\theta_{0}$,  $\Theta>1$  being specified later. The  sequence $\{R_l\}_{l=0}^\infty$  satisfies $R_l\searrow R_\infty\geq R_0/2.$ 
To clarify  the  iteration scheme, we set
 \begin{align*}
 	&\mathcal{R}_{0}=\mathcal{W}_{0}=\mathcal{M}_{-1}=0,\ V_{0}=V_{-1}=V,\\
	&\mathcal{M}_{l}=\mathcal{U}_{l}\mathcal{M}^{(l)}\mathcal{U}_{l}^{-1}+\mathcal{R}_{l}.
 \end{align*}
Equivalently, at the $l$-th iteration step we aim to  find $\mathcal{W}_{l}\in \mathscr{U}_{R_{l},s}^{\bm\omega}$, $V_l\in \mathscr{P}_{R_l}$ and $\mathcal{R}_{l}\in \mathscr{U}_{R_{l},s}^{\bm\omega}$ so that
\begin{align*}
	e^{\mathcal{W}_{l}}\left(V_{l-1}+\mathcal{M}_{l-1}\right)	e^{-\mathcal{W}_{l}}=V_{l}+\mathcal{R}_{l},
\end{align*}
with $\|\mathcal{R}_{l}\|_{R_{l},s}=o(\|\mathcal{R}_{l-1}\|_{R_{l-1},s})$. For this purpose, it needs to eliminate terms of order $O(\|\mathcal{R}_{l-1}\|_{R_{l-1},s})$,  which leads to solvling  the following homological equations
\begin{align*}
	&V_{l}(z)=V_{l-1}(z)+\mathcal{M}_{l-1}(z,\bm0),\ \mathcal{W}_{l}(z,\bm 0)=0,\\
	& (\mathcal{W}_l V_{l}-V_{l}\mathcal{W}_l)(z,\bm n)=-(S_{\theta_{l}}\tilde{\mathcal{M}}_{l-1})(z,\bm n)\ {\rm for}\ \bm n\neq\bm0,
\end{align*}
where $\tilde{\mathcal{M}}_{l-1}=(I-S_{0})\mathcal{M}_{l-1}$.

To solve the homological equations, we have to address the following issues. First, the non-resonant property  of $V_l$ should  be preserved,  which requires the use of  quantitative version of meromorphic monotone  function estimates  established by Bellissard-Lima-Scoppola \cite{BLS83}.  Second, since in the iteration steps we need to estimate the $\|\cdot\|_{R,s}$ norm of  products of elements in $\mathscr{U}_{R,s}^{\bm\omega}$,  the {tame} property  of the norm (cf. \eqref{ts})  plays an essential role.  

This section is then organized as follows. We first provide some useful estimates in \S  3.1.  The Nash-Moser iteration theorem is then  proved in \S 3.2. 

\subsection{Some useful estimates}
 Let $V'\in\mathscr{P}_{R'}$ and  $\mathcal{M}'\in\mathscr{U}_{R',s}^{\bm\omega}$.   
 We define $\bar{V'}$ and $\tilde{\mathcal{M}'}$  as follows
 \begin{align}
 	\label{vb}\bar{V'}(z)&=V'(z)+\mathcal{M}'(z,\bm0),\\
 	\nonumber\tilde{\mathcal{M}'}(z,\bm n)&=((I-S_{0})\mathcal{M}')(z,\bm n)=\left\{
 	\begin{array}{cc}
 		\mathcal{M}'(z,\bm n)& {\rm if} \ \bm n\ne \bm 0,\\
 		0 & {\rm if} \ \bm n= \bm 0. \\
 	\end{array}\right.
 \end{align}
 We have
 \begin{lem}\label{lem31}
 	If $R'>Q'>0$ is such that
 	\begin{align}\label{bm}
 		\|\mathcal{M}'\|_{R',0}<Q'|V'|_{R'},
 	\end{align}
 	then $\bar{V'}\in\mathscr{P}_{R'-Q'}$ and
 	\begin{align}\label{hv}
 		|\bar{V'}|_{R'-Q'}\ge |V'|_{R'}-\frac{\|\mathcal{M}'\|_{R',0}}{Q'}>0.
 	\end{align}
 \end{lem}
 \begin{proof}
 	The proof follows from
 	\begin{lem}[Lemma I.2, \cite{BLS83}]
 		Let $g$ be in $\mathscr{P}_{R'}$ and $f$ be in $\mathscr{H}_{R'}$. If $R'>Q'>0$ is such that
 		$\left \|f\right \|_{R'}<{Q'}|g|_{R'}$,
 		then $f+g\in\mathscr{P}_{R'-Q'}$ and
 		\begin{align*}
 			\Big|  |f+g|_{R'-Q'}-|g|_{R'-Q'}\Big|\le {(Q')}^{-1}\left \|f\right \|_{R'}.
 		\end{align*}
 	\end{lem}
 Hence it suffices to apply the above lemma with  $g(z)=V'(z)$ and $f(z)=\mathcal{M}'(z,\bm0)$.
 \end{proof}
For $\theta\ge0$, we define the  kernel $\mathcal{W}'$ as
\begin{align}\label{defw}
\mathcal{W}'(z,\bm 0)=0,\ \mathcal{W}'(z,\bm n)=\frac{(S_{\theta}\tilde{\mathcal{M}'})(z,\bm n)}{\bar{V'}(z)-\bar{V'}(z-\bm n\cdot\bm\omega)}\ {\rm for}\ \bm n\neq \bm0.
\end{align}
We have 
\begin{lem}\label{hollem}
\label{w*} For any $\theta\ge0$ and $\mathcal{W}'\in\mathscr{U}_{R'-Q',s}^{\bm\omega}$,  we have
\begin{align}\label{emw}
\left \|\mathcal{W}'\right \|_{R'-Q',s}\le \frac{\langle\theta\rangle^{\tau}}{\g |\bar{V'}|_{R'-Q'}}\|\mathcal{M}'\|_{R',s},
\end{align}
Moreover, if $(\mathcal{M}')^{*}=\mathcal{M}'$ and $(V')^{*}=V'$, then $(\mathcal{W}')^{*}=-\mathcal{W}'$.
\end{lem}

\begin{proof}
 Since \eqref{defw}, \eqref{sm1} and $\bm\omega\in {\rm DC}_{\tau,\g}$, we get
	\begin{align*}
	\|\mathcal{W}'\|_{R'-Q',s}=&\sup_{z\in{D}_{R'-Q'}}\sum_{\bm n\in\Z^{d}}|\mathcal{W}'(z,\bm n)|\langle\bm n\rangle^{s}\\
 \le& \frac{1}{\g |\bar{V'}|_{R'-Q'}}\sup_{z\in{D}_{R'}}\sum_{\bm n\in\Z^{d}}|(S_{\theta}\tilde{\mathcal{M}'})(z,\bm n)|\langle\bm n\rangle^{s+\tau}\\
	=&\frac{1}{\g |\bar{V'}|_{R'-Q'}}\|S_{\theta}\tilde{\mathcal{M}'}\|_{R',s+\tau}\le\frac{\langle\theta\rangle^{\tau}}{\g |\bar{V'}|_{R'-Q'}}\|\tilde{\mathcal{M}'}\|_{R',s}\\
	\le&\frac{\langle\theta\rangle^{\tau}}{\g |\bar{V}|_{R'-Q'}}\|\mathcal{M}'\|_{R',s},
	\end{align*}
	which  implies \eqref{emw}.

The last assertion  of the lemma follows directly from \eqref{*}.
\end{proof}
The following elementary inequality plays an important role in the proof of tame estimate. 
\begin{lem}\label{nl}
	Let $(x,y, s)\in \R_{+}^3\setminus\{(0,0,0)\}$. Then we have
		\begin{align}\label{xy}
		(x+y)^{s}\le K(s)(x^{s}+y^{s}),
	\end{align}
where
\begin{align}\label{cs}
	K(s)=2^{\max(0,s-1)}\ge1.
\end{align}
\end{lem}
We then introduce a key lemma concerning {tame} property. Recalling Lemma \ref{ts}, we have

\begin{lem}
	Let $K(s)$ be given by \eqref{cs}. Then for any $n\ge1$ and $s\ge0$, we have
	\begin{align}
	\label{s1}\left\|\prod_{i=1}^{n}\mathcal{N}_{i}\right\|_{R,0}\le&\prod_{i=1}^{n}\|\mathcal{N}_{i}\|_{R,0},\\
	\label{t1}\left\|\prod_{i=1}^{n}\mathcal{N}_{i}\right\|_{R,s}\le&(K(s))^{n-1}\sum_{i=1}^{n}\left(\prod_{j\ne i}\|\mathcal{N}_{j}\|_{R,0}\right)\|\mathcal{N}_{i}\|_{R,s}.
	\end{align}
	In particular,
	\begin{align}
	\label{s2}\|\mathcal{N}^{n}\|_{R,0}\le&\|\mathcal{N}\|_{R,0}^{n},\\
	\label{t2}\|\mathcal{N}^{n}\|_{R,s}\le&n(K(s))^{n-1}\|\mathcal{N}\|_{R,0}^{n-1}\|\mathcal{N}\|_{R,s}.
	\end{align}
\end{lem}
\begin{rem}
	In fact, imitating the proof in Lemma \ref{ts}, we can get a better estimate
	\begin{align*}
	\left\|\prod_{i=1}^{n}\mathcal{N}_{i}\right\|_{R,s}\le&K(s,n)\sum_{i=1}^{n}\left(\prod_{j\ne i}\|\mathcal{N}_{j}\|_{R,0}\right)\|\mathcal{N}_{i}\|_{R,s},
	\end{align*}
where $K(s,n)=\max\{1,n^{s-1}\}$. It may be useful elsewhere, but it is not necessary in this paper.
\end{rem}
	\begin{proof}
	The proof follows directly from an induction (on $n$) argument  using Lemma \ref{ts}. We refer to the proof of Lemma 4.2 in \cite{Shi22} for details. 		
	\end{proof}

Under the above preparations, we can prove
\begin{lem}
We have 
\begin{align*}
e^{\mathcal{W}'}(V'+\mathcal{M}')e^{-\mathcal{W}'}=\bar{V'}+\mathcal{R}',
\end{align*}
where  $\mathcal{R}'\in\mathscr{U}_{R'-Q',s}^{\bm\omega}$ for $\forall\ \theta\ge0$, and
\begin{align}
\nonumber&\ \ \ \ \|\mathcal{R}'\|_{R'-Q',s}\\
\nonumber&\le4K(s)e^{2K(s)\|\mathcal{W}'\|_{R'-Q',0}}\left(\|\mathcal{W}'\|_{R'-Q',0}\|\mathcal{M}'\|_{R',s}+\|\mathcal{M}'\|_{R',0}\|\mathcal{W}'\|_{R'-Q',s}\right)\\
\label{r}&\ \ +\|(I-S_{\theta})\mathcal{M}'\|_{R',s}.
\end{align}
\end{lem}
\begin{proof}
	We define for any $\mathcal{P}=(\mathcal{P}(z,\bm n))_{\bm n\in\Z^d, z\in {D}_{R}}$ and $k\geq 0,$
	\begin{align*}
	A_{\mathcal{W}'}^{k}(\mathcal{P})\equiv\sum_{i=0}^{k}\binom{k}{i}(\mathcal{W}')^{k-i}\mathcal{P}(-\mathcal{W}')^{i}.
	\end{align*}
	 Formally,   we have  by using  the Taylor series expansion
	\begin{align*}
	e^{\mathcal{W}'}\mathcal{P}e^{-\mathcal{W}'}=\sum_{k=0}^{\infty}\frac{A_{\mathcal{W}'}^{k}(\mathcal{P})}{k!}.
	\end{align*}
From \eqref{defw},  we can obtain
\begin{align*}
\mathcal{W}'\bar{V'}-\bar{V'}\mathcal{W}'=-S_{\theta}\tilde{\mathcal{M}'},
\end{align*}
and then for $k\geq1$,
	\begin{align*}
	A_{\mathcal{W}'}^{k}(\bar{V})=A_{\mathcal{W}'}^{k-1}(-S_{\theta}\tilde{\mathcal{M}'})=-A_{\mathcal{W}'}^{k-1}(S_{\theta}\tilde{\mathcal{M}'}).
	\end{align*}
As a result,  we get
	\begin{align}
	\nonumber &\ \ \ \ e^{\mathcal{W}'}(V'+\mathcal{M}')e^{-\mathcal{W}'}=e^{\mathcal{W}'}(\bar{V'}+\tilde{\mathcal{M}'})e^{-\mathcal{W}'}\\
	\nonumber&=\bar{V'}+\sum_{k=2}^{\infty}\frac{A_{\mathcal{W}'}^{k}(\bar{V'})}{k!}+\sum_{k=1}^{\infty}\frac{A^{k}_{\mathcal{W}'}(S_{\theta}\tilde{\mathcal{M}'})}{k!}+\sum_{k=0}^{\infty}\frac{A_{\mathcal{W}'}^{k}((I-S_{\theta})\tilde{\mathcal{M}'})}{k!}\\
	\nonumber&=\bar{V'}+\sum_{k=1}^{\infty}\frac{A_{\mathcal{W}'}^{k}(S_{\theta}\tilde{\mathcal{M}'})}{(k-1)!(k+1)}+\sum_{k=0}^{\infty}\frac{A_{\mathcal{W}'}^{k}((I-S_{\theta})\tilde{\mathcal{M}'})}{k!}\\
	&=\bar{V'}+\mathcal{R'}.\label{vmbar}
	\end{align}

	Next, we try to establish \eqref{r}. Using \eqref{t1} yields for $k\geq 1,$
	\begin{align*}
	&\ \ \ \ \left \|A_{\mathcal{W}'}^{k}(\mathcal{P})\right \|_{R'-Q',s}\\
	&\le \sum_{i=0}^{k}\left \|\binom{k}{i}(\mathcal{W}')^{k-i}\mathcal{P}(-\mathcal{W}')^{i}\right \|_{R'-Q',s}\\
	&\le \sum_{i=0}^{k}\binom{k}{i}(K(s))^{k}\left(\|\mathcal{W}'\|_{R'-Q',0}^{k}\|\mathcal{P}\|_{R'-Q',s}+\|\mathcal{W}'\|_{R'-Q',0}^{k-1}\|\mathcal{P}\|_{R'-Q',0}\|\mathcal{W}'\|_{R'-Q',s}\right)\\
	&=2^k(K(s))^{k}\left(\|\mathcal{W}'\|_{R'-Q',0}^{k}\|\mathcal{P}\|_{R'-Q',s}+\|\mathcal{W}'\|_{R'-Q',0}^{k-1}\|\mathcal{P}\|_{R'-Q',0}\|\mathcal{W}'\|_{R'-Q',s}\right).
	\end{align*}
Therefore, we have
\begin{align}
\nonumber&\ \ \ \  \left\|\sum_{k=1}^{\infty}\frac{A_{\mathcal{W}'}^{k}(S_{\theta}\tilde{\mathcal{M}'})}{(k-1)!(k+1)}\right\|_{R'-Q',s}\\
\nonumber\le&2K(s)\left(\|\mathcal{W}'\|_{R'-Q',0}\|S_{\theta}\tilde{\mathcal{M}'}\|_{R'-Q',s}+\|S_{\theta}\tilde{\mathcal{M}'}\|_{R'-Q',0}\|\mathcal{W}'\|_{R'-Q',s}\right)\\
\nonumber&\ \ \times\sum_{k=1}^{\infty}\frac{\left(2K(s)\right)^{k-1}\left \|\mathcal{W}' \right \|_{R'-Q',0}^{k-1}}{(k-1)!}\\
\nonumber\le&2K(s)e^{2K(s)\|\mathcal{W}'\|_{R'-Q',0}}\left(\|\mathcal{W}'\|_{R'-Q',0}\|S_{\theta}\tilde{\mathcal{M}'}\|_{R'-Q',s}+\|S_{\theta}\tilde{\mathcal{M}'}\|_{R'-Q',0}\|\mathcal{W}'\|_{R'-Q',s}\right)\\
\label{asm}\le&2K(s)e^{2K(s)\|\mathcal{W}'\|_{R'-Q',0}}\left(\|\mathcal{W}'\|_{R'-Q',0}\|S_{\theta}\mathcal{M}'\|_{R',s}+\|S_{\theta}\mathcal{M}'\|_{R',0}\|\mathcal{W}'\|_{R'-Q',s}\right).
\end{align}
Similarly, we get
\begin{align}
\nonumber&\ \ \ \ \left\|\sum_{k=0}^{\infty}\frac{A_{\mathcal{W}'}^{k}((I-S_{\theta})\tilde{\mathcal{M}'})}{k!}\right\|_{R'-Q',s}\\
\nonumber&\le2K(s)\|(I-S_{\theta})\tilde{\mathcal{M}'}\|_{R'-Q',s}\|\mathcal{W}'\|_{R'-Q',0}\sum_{k=1}^{\infty}\frac{2^{k-1}(K(s))^{k-1}\|\mathcal{W}'\|_{R'-Q',0}^{k-1}}{(k-1)!}\\
&\ \ +\nonumber 2K(s)\|(I-S_{\theta})\tilde{\mathcal{M}'}\|_{R'-Q',0}\|\mathcal{W}'\|_{R'-Q',s}\sum_{k=1}^{\infty}\frac{2^{k-1}(K(s))^{k-1}\|\mathcal{W}'\|_{R'-Q',0}^{k-1}}{(k-1)!}\\
\nonumber&\ \  +\|(I-S_{\theta})\tilde{\mathcal{M}'}\|_{R'-Q',s}\\
\nonumber&\le2K(s)e^{2K(s)\|\mathcal{W}'\|_{R'-Q',0}}\left(\|(I-S_{\theta})\mathcal{M}'\|_{R',s}\|\mathcal{W}'\|_{R'-Q',0}+\|(I-S_{\theta})\mathcal{M}'\|_{R',0}\|\mathcal{W}'\|_{R'-Q',s}\right)\\
\label{a1sm}&\  \ +\|(I-S_{\theta})\mathcal{M}'\|_{R',s}.
\end{align}
Recalling the Definition \ref{sm}  and \eqref{norm}, we obtain
\begin{align*}
\|S_{\theta}\mathcal{M}'\|_{R',s}+\|(I-S_{\theta})\mathcal{M}'\|_{R',s}&\le2\|\mathcal{M}'\|_{R',s},\\
\|S_{\theta}\mathcal{M}'\|_{R',0}+\|(I-S_{\theta})\mathcal{M}'\|_{R',0}&\le2\|\mathcal{M}'\|_{R',0},
\end{align*}
which together with \eqref{vmbar},  \eqref{asm} and \eqref{a1sm} implies
	\begin{align*}
	&\ \ \ \  \|\mathcal{R}'\|_{R'-Q',s}\\
	&\le\left\|\sum_{k=1}^{\infty}\frac{A_{\mathcal{W}'}^{k}(S_{\theta}\tilde{\mathcal{M}'})}{(k-1)!(k+1)}\right\|_{R'-Q',s}+\left\|\sum_{k=0}^{\infty}\frac{A_{\mathcal{W}'}^{k}((I-S_{\theta})\tilde{\mathcal{M}'})}{k!}\right\|_{R'-Q',s}\\
	&\le 2K(s)e^{2K(s)\|\mathcal{W}'\|_{R'-Q',0}}\left(\|\mathcal{W}'\|_{R'-Q',0}\|S_{\theta}\mathcal{M}'\|_{R',s}+\|S_{\theta}\mathcal{M}'\|_{R',0}\|\mathcal{W}'\|_{R'-Q',s}\right)\\
	&\ \ +2K(s)e^{2K(s)\|\mathcal{W}'\|_{R'-Q',0}}\left(\|(I-S_{\theta})\mathcal{M}'\|_{R',s}\|\mathcal{W}'\|_{R'-Q',0}+\|(I-S_{\theta})\mathcal{M}'\|_{R',0}\|\mathcal{W}'\|_{R'-Q',s}\right)\\
	&\ \  +\|(I-S_{\theta})\mathcal{M}'\|_{R',s}\\
	&\leq4K(s)e^{2K(s)\|\mathcal{W}'\|_{R'-Q',0}}\left(\|\mathcal{W}'\|_{R'-Q',0}\|\mathcal{M}'\|_{R',s}+\|\mathcal{M}'\|_{R',0}\|\mathcal{W}'\|_{R'-Q',s}\right)\\
	&\ \ +\|(I-S_{\theta})\mathcal{M}'\|_{R',s}.
	\end{align*}

This completes the proof. 
\end{proof}

\subsection{The Nash-Moser iteration theorem}
In this subsection we try to establish the iteration theorem.  We first introduce some parameters.
\begin{itemize}
		\item Fix  $\delta>0$ and $\alpha_0>0$.  Let
	\begin{align}\label{alpha}
		\alpha>\alpha_{0}+\tau+4\delta.
	\end{align}
	\item Fix any  $\alpha_{1}=\alpha_1(\alpha,\delta)>0$  so that 
	\begin{align}\label{0}
		\alpha_{1}>2\alpha+\delta. 
	\end{align}
	\item Let $\Theta=\Theta(\alpha_{0},\alpha, \delta)>0$ satisfy 
	\begin{align}
	\label{sml1}\Theta^{-\delta}&\le\frac{1}{4}e^{-2K(\alpha_{1})}\le\frac{1}{4},\\
	\label{sml2}\Theta^{-\alpha_{0}}&\le\frac{1}{4}.
	\end{align}
     \item Let $\eta_{0}=\eta_{0}(R,\Theta,\alpha,\alpha_{0}, \g,\tau,|V|_{R},\delta)>0$ be the minimum value of $\eta$  satisfying 
     \begin{align}
     \label{tv}\eta^{\frac{\alpha_{0}-\alpha}{2}}&\le\frac{R|V|_{R}}{8}\left(1-\Theta^{\frac{\alpha_{0}-\alpha}{2}}\right),\ \eta^{\frac{\alpha_{0}-\alpha}{2}}\le1-\Theta^{\frac{\alpha_{0}-\alpha}{2}},\\
     	\label{Tg}\eta^{-\delta}&\le\frac{\g|V|_{R}}{4\Theta^{\tau}},\ \eta^{-\delta}\le(32K(\alpha_{1}))^{-1}e^{-2K(\alpha_{1})}\Theta^{\alpha_0-\alpha},\\
     	\label{tT}\eta^{-1}&\le\Theta^{-1},\ \eta^{\alpha_{0}-\alpha}\le\Theta^{\alpha_{0}-\alpha-3\delta},\\
     	\label{12}\eta^{-\delta}&\le(12(K(\alpha_{1}))^{2})^{-1}.
     \end{align}
	\item Let $\ep_{0}=\ep_{0}(R,\alpha,\alpha_{0},\g,\tau,|V|_{R},\delta)>0$  satisfy
	\begin{align}\label{e0}
	\ep_{0}\le\eta_{0}^{\alpha_{0}-\alpha}\le1.
	\end{align}
		\item Let $\theta_{0}\ge\eta_{0}$ and $\theta_{l}=\theta_{0}\Theta^{l}$. Denote by  $\mathcal{M}^{(l)}$ the $l$-section of $\mathcal{M}$ w.r.t  $\{\theta_{l}\}_{l=0}^{\infty}$.
\item Let
\begin{align}
	\label{ql}R_{0}&=R,\ Q_{l}=\frac{4}{|V|_{R}}\theta_{l}^{\frac{\alpha_{0}-\alpha}{2}},\\
	\nonumber R_{l+1}&=R_{l}-Q_{l}\ge\frac{R}{2}\left(1+\Theta^{(\frac{\alpha_{0}-\alpha}{2})(l+1)}\right) ({\rm since}\  \eqref{tv}).
\end{align}
\end{itemize}
We start with a useful lemma. 
\begin{lem}\label{ew}
 For all $\mathcal{N}_{m}\in\mathscr{U}_{R',s}^{\bm\omega}$ with $1\le m\le k$, we have
 \begin{align}\label{Q}
 \left\|e^{\mathcal{N}_{1}}\cdots e^{\mathcal{N}_{k}}-\bm 1\right\|_{R',s}\le e^{K(s)\left(\sum\limits_{m=1}^{k}\|\mathcal{N}_{m}\|_{R',0}\right)}\left(\sum_{m=1}^{k}\|\mathcal{N}_{m}\|_{R',s}\right).
 \end{align}
\end{lem}
\begin{proof}
We refer to the appendix for a detailed proof.
\end{proof}
We are able to state our iteration theorem.
\begin{thm}\label{main}
	If $\left \|\mathcal{M} \right \|_{R,\alpha+3\delta}<\ep_{0}$, then there exists a sequence
	\begin{align*}
	\left(V_{l},\mathcal{M}_{l},\mathcal{W}_{l+1},\mathcal{R}_{l+1}\right)_{l=0}^{\infty}\in\mathscr{P}_{R_{l}}\times\mathscr{U}_{R_{l},s}^{\bm\omega}\times\mathscr{U}_{R_{l+1},s}^{\bm\omega}\times\mathscr{U}_{R_{l+1},s}^{\bm\omega}\ (s\in[\alpha_0, \alpha_1])
	\end{align*}
	satisfying
	\begin{align}
	\nonumber &e^{\mathcal{W}_{l+1}}(V_{l}+\mathcal{M}_{l})e^{-\mathcal{W}_{l+1}}=V_{l+1}+\mathcal{R}_{l+1},\\
	\label{u}&\mathcal{U}_{l+1}=\prod_{i=l+1}^{1}e^{\mathcal{W}_{i}},\ V_{0}=V,\ \mathcal{M}_{0}=\mathcal{M}^{(0)},\\
	\label{vm}&V_{l+1}=\bar{V}_{l},\ \mathcal{M}_{l+1}=\mathcal{U}_{l+1}\mathcal{M}^{(l+1)}\mathcal{U}_{l+1}^{-1}+\mathcal{R}_{l+1},
\end{align}
	so that
	\begin{align}
		\label{ml}\|\mathcal{M}_{l}\|_{R_{l},s}&\le2\theta_{l}^{s-\alpha}\  {\rm for}\ s\in[\alpha_{0},\alpha_{1}],\\ \label{vl}|V_{l}|_{R_{l}}&\ge|V|_{R}-\sum_{j=0}^{l-1}\frac{\|\mathcal{M}_{j}\|_{R_{j},0}}{Q_{j}}\ge\frac{|V|_{R}}{2} \ (l\ge1),\\
	\label{wl}\|\mathcal{W}_{l+1}\|_{R_{l+1},s}&\le\theta_{l}^{s-\alpha+\tau+\delta}\ {\rm for}\ s\in[\alpha_{0},\alpha_{1}],\\
	\label{ul}\|\mathcal{U}_{l+1}-\bm1\|_{R_{l+1},s}&\le\theta_{l}^{(s-\alpha+\tau+\delta)_{+}+\delta}\ {\rm for}\ s\in[\alpha_{0},\alpha_{1}],\\
	\nonumber\|\mathcal{R}_{l+1}\|_{R_{l+1},s}&\le\theta_{l+1}^{s-\alpha}\  {\rm for}\ s\in[\alpha_{0},\alpha_{1}],
	\end{align}
	where $x_{+}=x$ if $x\ge0$ and $x_{+}=0$ if $x<0$. In addition, if  both $\mathcal{M}$ and $V$ are self-adjoint, then for each $l\ge0$, $\mathcal{U}_{l+1}$ is unitary and $V_{l}^{*}=V_{l}$.
\end{thm}

\begin{proof}[Proof of Theorem \ref{main}]
	We first check \eqref{ml} and \eqref{vl} hold true for $\mathcal{M}_{0}$ and $V_0$ respectively. Let $V$ and  $\mathcal{M}$ be as in Theorem \ref{mthm}. Since {\color{red}{\eqref{sm1}}}, \eqref{e0} and $\|\mathcal{M}\|_{R,\alpha+3\delta}<\ep_{0}$, we obtain for $s\ge\alpha+3\delta$,
	\begin{align}\label{m0s}
		\|\mathcal{M}_{0}\|_{R,s}&\le\theta_{0}^{s-\alpha-3\delta}\|\mathcal{M}\|_{R,\alpha+3\delta}\le\theta_{0}^{s-\alpha-3\delta}.
	\end{align}
If $\alpha_{0}\le s<\alpha+3\delta$, we  have
\begin{align*}
	\|\mathcal{M}_{0}\|_{R,s}\le\|\mathcal{M}\|_{R,\alpha+3\delta}\le\theta_{0}^{\alpha_{0}-\alpha}\le\theta_{0}^{s-\alpha},
\end{align*}
which combined with \eqref{m0s} implies
\begin{align*}
	\|\mathcal{M}_{0}\|_{R,s}\le2\theta_{0}^{s-\alpha}\ \text{for}\ s\in[\alpha_{0},\alpha_{1}].
\end{align*}
Obviously, we have $|V_0|_{R_0}=|V|_{R}\ge\frac{1}{2}|V|_{R}.$

Next,  assume that  for $0\leq l\le L$, we have  constructed $V_{l},\mathcal{M}_{l}$ using \eqref{u} and \eqref{vm} so that both \eqref{vl} and \eqref{ml}  hold true. We want to use \eqref{u} and \eqref{vm} to construct $V_{L+1}$, $\mathcal{M}_{L+1}$  so that  \eqref{vl} and \eqref{ml}  hold true again.  Recalling \eqref{vb}, we let
	\begin{align*}
		V_{L+1}(z)=\bar{V}_{L}(z)=V_{L}(z)+\mathcal{M}_{L}(z,\bm0).
	\end{align*}
According to {\color{red}\eqref{alpha}}, \eqref{ql}, \eqref{ml} and \eqref{vl}, we get
	\begin{align*}
	\|\mathcal{M}_{L}\|_{R_L,0}&\le\|\mathcal{M}_{L}\|_{R_L,\alpha_{0}}\le2\theta_{L}^{\alpha_{0}-\alpha}
< Q_{L}|V|_{R_L},
	\end{align*}
    which together with Lemma \ref{lem31}, \eqref{tv} and \eqref{vl} implies
    \begin{align}
    \nonumber|V_{L+1}|_{R_{L+1}}&\ge|V_L|_{R_L}-\frac{\|\mathcal{M}_{L}\|_{R_{L},0}}{Q_{L}}\ge|V|_{R}-\sum_{j=0}^{L}\frac{\|\mathcal{M}_{j}\|_{R_{j},0}}{Q_{j}}\\
    \nonumber&\ge|V|_{R}-\frac{|V|_{R}}{2}\sum_{j=0}^{\infty}\theta_{j}^{\frac{\alpha_{0}-\alpha}{2}}\\
    \label{vl+1}&=|V|_{R}-\frac{|V|_{R}}{2}\frac{\theta_{0}^{\frac{\alpha_{0}-\alpha}{2}}}{1-\Theta^{\frac{\alpha_{0}-\alpha}{2}}}\ge\frac{|V|_{R}}{2}.
    \end{align}
	Next, $\mathcal{W}_{L+1}$ is obtained by setting $\bar{V'}=V_{L+1}, \mathcal{M}'=\mathcal{M}_{L}$ and $\theta=\theta_{L+1}$ via \eqref{defw}. Since \eqref{emw}, \eqref{Tg}, \eqref{ml} and \eqref{vl+1}, we get
	\begin{align*}
		\|\mathcal{W}_{L+1}\|_{R_{L+1},s}&\le\frac{\theta_{L+1}^{\tau}}{\g|V_{L+1}|_{R_{L+1}}}\left \|\mathcal{M}_{L}\right \|_{R_{L},s}\\
		&\le\frac{4\Theta^{\tau}}{\g|V|_{R}}\theta_{0}^{-\delta}\theta_{L}^{s-\alpha+\tau+\delta}\\
		&\le\theta_{L}^{s-\alpha+\tau+\delta}\ \text{for}\ s\in[\alpha_{0},\alpha_{1}].
	\end{align*}
	 
	 In the following we provide corresponding estimates. To estimate $\|\mathcal{R}_{l+1}\|_{R_{l+1},s}$,  we first deal with $\|(I-S_{\theta_{L+1}})\mathcal{M}_{L}\|_{R_{L},s}$. We have two cases.
	\begin{itemize}
		\item [\textbf{Case 1.}]$s\in[\alpha+\delta,\alpha_{1}]$.  From  \eqref{sml1} and \eqref{ml}, we have
		\begin{align*}
		\|(I-S_{\theta_{L+1}})\mathcal{M}_{L}\|_{R_{L},s}\le&\|\mathcal{M}_{L}\|_{R_{L},s}\le2\theta_{L}^{s-\alpha}\\
		\le&2\Theta^{\alpha-s}\theta_{L+1}^{s-\alpha}\le2\Theta^{-\delta}\theta_{L+1}^{s-\alpha}\\
		\le&\frac{1}{2}\theta_{L+1}^{s-\alpha}.
		\end{align*}
		\item[\textbf{Case 2.}] $s\in[\alpha_{0},\alpha+\delta)$. Since \eqref{sm2},  \eqref{0},  \eqref{sml2} and \eqref{ml}, we have
		\begin{align*}
		\|(I-S_{\theta_{L+1}})\mathcal{M}_{L}\|_{R_{L},s}\le&\theta_{L+1}^{-\alpha}\|\mathcal{M}_{L}\|_{R_{L},s+\alpha}\\
		\le&2\theta_{L+1}^{-\alpha}\theta_{L}^{s}\le2\Theta^{-\alpha_{0}}\theta_{L+1}^{s-\alpha}\\
		\le&\frac{1}{2}\theta_{L+1}^{s-\alpha}.
		\end{align*}
	\end{itemize}
	To sum up,  one has
	\begin{align}\label{a+d}
		\|(I-S_{\theta_{L+1}})\mathcal{M}_{L}\|_{R_{L},s}\le\frac{1}{2}\theta_{L+1}^{s-\alpha}\ \text{for}\ s\in[\alpha_{0},\alpha_{1}].
	\end{align}
	Hence from \eqref{r}, \eqref{alpha}, \eqref{Tg}, \eqref{wl}, \eqref{ml} and \eqref{a+d}, we have
	\begin{align}
	\nonumber&\ \ \ \ \|\mathcal{R}_{L+1}\|_{R_{L+1},s}\\
	\nonumber&\le4K(s)e^{2K(s)\|\mathcal{W}_{L+1}\|_{R_{L+1},0}}\left(\|\mathcal{W}_{L+1}\|_{R_{L+1},0}\|\mathcal{M}_L\|_{R_L,s}+\|\mathcal{M}_L\|_{R_L,0}\|\mathcal{W}_{L+1}\|_{R_{L+1},s}\right)\\
	\nonumber&\ \ +\|(I-S_{\theta_{L+1}})\mathcal{M}_L\|_{R_L,s}\\
	\nonumber&\le16K(\alpha_{1})e^{2K(\alpha_{1})}\theta_{L}^{s-2\alpha+\alpha_{0}+\tau+\delta}+\frac{1}{2}\theta_{L+1}^{s-\alpha}\\
	\nonumber&\le16K(\alpha_{1})e^{2K(\alpha_{1})}\theta_{0}^{-\delta}\theta_{L}^{s-\alpha}+\frac{1}{2}\theta_{L+1}^{s-\alpha}\\
	\label{rl}&\le\theta_{L+1}^{s-\alpha}\ \text{for}\ s\in[\alpha_{0},\alpha_{1}].
	\end{align}
	By {\eqref{Q} and \eqref{wl}}, we can obtain
	\begin{align}
	\nonumber\|\mathcal{U}_{L+1}-\bm1\|_{R_{L+1},s}&\le e^{K(s)\left(\sum\limits_{j=0}^{L}\|\mathcal{W}_{j+1}\|_{R_{j+1},0}\right)}\left(\sum_{j=0}^{L}\|\mathcal{W}_{j+1}\|_{R_{j+1},s}\right)\\
	\nonumber&\le e^{K(s)\left(\sum\limits_{j=0}^{L}\|\mathcal{W}_{j+1}\|_{R_{j+1},\alpha_{0}}\right)}\left(\sum_{j=0}^{L}\|\mathcal{W}_{j+1}\|_{R_{j+1},s}\right)\\
	\label{ul1}&\le(L+1)e^{(L+1)K(\alpha_{1})}\underset{0\le j\le L}{\max}\|\mathcal{W}_{j+1}\|_{R_{j+1},s}.
	\end{align}
	Since \eqref{sml1} and $L+1\le2^{L+1}$, we have
	\begin{align*}
	(L+1)e^{(L+1)K(\alpha_{1})}&\le\left(2e^{K(\alpha_{1})}\right)^{L+1}\le\theta_{0}^{\frac{\delta}{2}}(\Theta^{\frac{\delta}{2}})^{L+1}=\theta_{L+1}^{\frac{\delta}{2}}.
	\end{align*}
	Together with \eqref{tT}, \eqref{wl} and \eqref{ul1}, we can obtain for $\alpha_{0}\le s<\alpha-\tau-\delta$ that
	\begin{align}\label{q1}
	 \|\mathcal{U}_{L+1}-\bm1\|_{R_{L+1},s}\le\theta_{L+1}^{\frac{\delta}{2}}\theta_{0}^{s-\alpha+\tau+\delta}
	\le\theta_{L}^{\delta}\theta_{0}^{-\frac{\delta}{2}}\Theta^{\frac{\delta}{2}}
	\le\theta_{L}^{\delta}.
	\end{align}
	Similarly, for $\alpha-\tau-\delta\le s\le\alpha_{1}$, we obtain
	\begin{align}\label{q2}
	\|\mathcal{U}_{L+1}-\bm1\|_{R_{L+1},s}\le\theta_{L+1}^{\frac{\delta}{2}}\theta_{L}^{s-\alpha+\tau+\delta}
	\le\theta_{0}^{-\frac{\delta}{2}}\Theta^{\frac{\delta}{2}}\theta_{L}^{s-\alpha+\tau+2\delta}
	\le\theta_{L}^{s-\alpha+\tau+2\delta}.
	\end{align}
	Combining \eqref{q1} and \eqref{q2} implies \eqref{ul}. Hence
	\begin{align}\label{u1l}
	\|\mathcal{U}_{L+1}\|_{R_{L+1},s}&\le1+\theta_{L}^{(s-\alpha+\tau+\delta)_{+}+\delta}
	\le2\theta_{L}^{(s-\alpha+\tau+\delta)_{+}+\delta}.
	\end{align}
In the same way, we can get
\begin{align}
		\label{u2l}\|\mathcal{U}_{L+1}^{-1}\|_{R_{L+1},s}\le2\theta_{L}^{(s-\alpha+\tau+\delta)_{+}+\delta}.
\end{align}
Similar to \eqref{m0s}, combining \eqref{m1l} and \eqref{e0}, we have for $s\ge\alpha+3\delta$,
\begin{align}
	\nonumber\|\mathcal{M}^{(L+1)}\|_{R_{L+1},s}&\le\theta_{L+1}^{s-\alpha-3\delta}\|\mathcal{M}\|_{R_{L+1},\alpha+3\delta}\le\theta_{L+1}^{s-\alpha-3\delta}\|\mathcal{M}\|_{R,\alpha+3\delta}\\
	\label{ml1}&\le\theta_{L+1}^{s-\alpha-3\delta}.
\end{align}
From \eqref{m2l}, \eqref{tT} and \eqref{e0},  we obtain for $\alpha_{0}\le s<\alpha+3\delta$, 
\begin{align}
	\nonumber\|\mathcal{M}^{(L+1)}\|_{R_{L+1},s}&\le\theta_{L}^{s-\alpha-3\delta}\|\mathcal{M}\|_{R_{L+1},\alpha+3\delta}
	\le\theta_{L}^{s-\alpha-3\delta}\|\mathcal{M}\|_{R,\alpha+3\delta}\\
	\nonumber&\le\theta_{L+1}^{s-\alpha-3\delta}\Theta^{\alpha+3\delta-\alpha_{0}}\theta_{0}^{\alpha_{0}-\alpha}\\
	\label{ml2}&\le\theta_{L+1}^{s-\alpha-3\delta}.
\end{align}
Combining \eqref{ml1} and \eqref{ml2} yields
\begin{align}
	\label{ml3}\|\mathcal{M}^{(L+1)}\|_{R_{L+1},s}\le\theta_{L+1}^{s-\alpha-3\delta}\ \text{for}\ s\in[\alpha_{0},\alpha_{1}].
\end{align}
If $\alpha_{0}\le s<\alpha-\tau-\delta$,  then
\begin{align*}
(s-\alpha+\tau+\delta)_{+}+\alpha_{0}-\alpha-\delta&=\alpha_{0}-\alpha-\delta\le s-\alpha-\delta.
\end{align*}
If $\alpha-\tau-\delta\le s\le \alpha_{1}$, then
\begin{align*}
	(s-\alpha+\tau+\delta)_{+}+\alpha_{0}-\alpha-\delta&\le s-\alpha-\delta+(\alpha_{0}-\alpha+\tau+\delta)\\
	&\le s-\alpha-\delta\ {(\text{since}\ \eqref{alpha})}.
\end{align*}
As a result, we have
\begin{align}
	\label{+}(s-\alpha+\tau+\delta)_{+}+\alpha_{0}-\alpha-\delta\le s-\alpha-\delta\ \text{for}\ s\in[\alpha_{0},\alpha_{1}].
\end{align}
Since \eqref{u1l}, \eqref{u2l}, \eqref{ml3} and \eqref{+}, we get
\begin{align*}
	&\ \ \ \ \|\mathcal{U}_{L+1}\|_{R_{L+1},0}\|\mathcal{M}^{(L+1)}\|_{R_{L+1},0}\|\mathcal{U}_{L+1}^{-1}\|_{R_{L+1},s}\\
	&\le4\theta_{L+1}^{\alpha_{0}-\alpha-3\delta}\theta_{L}^{(s-\alpha+\tau+\delta)_{+}+2\delta}\le4\theta_{L+1}^{(s-\alpha+\tau+\delta)_{+}+\alpha_{0}-\alpha-\delta}\\
	&\le4\theta_{L+1}^{s-\alpha-\delta},\\
	&\ \ \ \  \|\mathcal{U}_{L+1}\|_{R_{L+1},s}\|\mathcal{M}^{(L+1)}\|_{R_{L+1},0}\|\mathcal{U}_{L+1}^{-1}\|_{R_{L+1},s} \\
	&       \le4\theta_{L+1}^{s-\alpha-\delta}, 
\end{align*}
and
\begin{align*}
	&\ \ \ \ \|\mathcal{U}_{L+1}\|_{R_{L+1},0}\|\mathcal{M}^{(L+1)}\|_{R_{L+1},s}\|\mathcal{U}_{L+1}^{-1}\|_{R_{L+1},0}\\
	&\le4\theta_{L+1}^{s-\alpha-3\delta}\theta_{L}^{2\delta}\le4\theta_{L+1}^{s-\alpha-\delta}.
\end{align*}
According to \eqref{t1},  we obtain  for $s\in[\alpha_0, \alpha_1],$
\begin{align*}
	\|\mathcal{U}_{L+1}\mathcal{M}^{(L+1)}\mathcal{U}_{L+1}^{-1}\|_{R_{L+1},s}&\le (K(s))^{2}\|\mathcal{U}_{L+1}\|_{R_{L+1},0}\|\mathcal{M}^{(L+1)}\|_{R_{L+1},0}\|\mathcal{U}_{L+1}^{-1}\|_{R_{L+1},s}\\
	&\ \  +(K(s))^{2}\|\mathcal{U}_{L+1}\|_{R_{L+1},0}\|\mathcal{M}^{(L+1)}\|_{R_{L+1},s}\|\mathcal{U}_{L+1}^{-1}\|_{R_{L+1},0}\\
	&\ \  +(K(s))^{2}\|\mathcal{U}_{L+1}\|_{R_{L+1},s}\|\mathcal{M}^{(L+1)}\|_{R_{L+1},0}\|\mathcal{U}_{L+1}^{-1}\|_{R_{L+1},s}\\
	&\le12(K(\alpha_{1}))^{2}\theta_{L+1}^{s-\alpha-\delta}.
\end{align*}
Thus by combining {\eqref{12} and \eqref{rl}}, we have
\begin{align*}
\|\mathcal{M}_{L+1}\|_{R_{L+1},s}&\le\|\mathcal{U}_{L+1}\mathcal{M}^{(L+1)}\mathcal{U}_{L+1}^{-1}\|_{R_{L+1},s}+\|\mathcal{R}_{L+1}\|_{R_{L+1},s}\\
&\le12(K(\alpha_{1}))^{2}\theta_{L+1}^{s-\alpha-\delta}+\theta_{L+1}^{s-\alpha}\\
&\le12(K(\alpha_{1}))^{2}\theta_{0}^{-\delta}\theta_{L+1}^{s-\alpha}+\theta_{L+1}^{s-\alpha}\\
&\le2\theta_{L+1}^{s-\alpha}.
\end{align*}

Finally, we prove for each $l\ge0$, $\mathcal{U}_{l+1}$ is unitary and $V_{l}^{*}=V_{l}$ by induction.  For $l=0$, since  $V^{*}=V$ and $\mathcal{M}^{*}=\mathcal{M}$, we have $\mathcal{M}_{0}^{*}=\mathcal{M}_{0}$ and $V_{1}^{*}=V_{1}$,
which combined with Lemma \ref{w*} implies $\mathcal{W}_{1}^{*}=-\mathcal{W}_{1}$. Thus
\begin{align*}
	\mathcal{R}_{1}^{*}&=\left(\sum_{k=1}^{\infty}\frac{A_{\mathcal{W}_{1}}^{k}(S_{\theta}\tilde{\mathcal{M}}_{0})}{(k-1)!(k+1)}+\sum_{k=0}^{\infty}\frac{A_{\mathcal{W}_{1}}^{k}((I-S_{\theta})\tilde{\mathcal{M}}_{0})}{k!}\right)^{*}\\
	&=\sum_{k=1}^{\infty}\frac{A_{\mathcal{W}_{1}}^{k}(S_{\theta}\tilde{\mathcal{M}}_{0}^{*})}{(k-1)!(k+1)}+\sum_{k=0}^{\infty}\frac{A_{\mathcal{W}_{1}}^{k}((I-S_{\theta})\tilde{\mathcal{M}}_{0}^{*})}{k!}\\
	&=\sum_{k=1}^{\infty}\frac{A_{\mathcal{W}_{1}}^{k}(S_{\theta}\tilde{\mathcal{M}}_{0})}{(k-1)!(k+1)}+\sum_{k=0}^{\infty}\frac{A_{\mathcal{W}_{1}}^{k}((I-S_{\theta})\tilde{\mathcal{M}}_{0})}{k!}=\mathcal{R}_{1}
\end{align*}
and
\begin{align*}
	&(e^{\mathcal{W}_{1}})^{*}e^{\mathcal{W}_{1}}=e^{-\mathcal{W}_{1}}e^{\mathcal{W}_{1}}=\bm1,\\
	&e^{\mathcal{W}_{1}}(e^{\mathcal{W}_{1}})^{*}=e^{\mathcal{W}_{1}}e^{-\mathcal{W}_{1}}=\bm1.
\end{align*}
Therefore,
\begin{align*}
	\mathcal{M}_{1}^{*}
	&=\left(e^{\mathcal{W}_{1}}\mathcal{M}^{(1)}e^{-\mathcal{W}_{1}}+\mathcal{R}_{1}\right)^{*}\\
	&=e^{\mathcal{W}_{1}}\mathcal{M}^{(1)}e^{-\mathcal{W}_{1}}+\mathcal{R}_{1}\\
	&=\mathcal{U}_{1}\mathcal{M}^{(1)}\mathcal{U}_{1}^{-1}+\mathcal{R}_{1}=	\mathcal{M}_{1}.	
\end{align*}
Now, we assume $V_{L}^{*}=V_{L}, \mathcal{M}_{L}^{*}=\mathcal{M}_{L}, \mathcal{R}_{L}^{*}=\mathcal{R}_{L}$ and $\mathcal{U}_{L}$ is unitary for  $L\ge1$.
Recalling \eqref{vm}, we have
\begin{align*}
	V_{L+1}^{*}&=V_{L}^{*}+\mathcal{M}_{L}^{*}(z,\bm0)=V_{L}+\mathcal{M}_{L}(z,\bm0)=V_{L+1}.
\end{align*}
Similar to the above arguments, we have
\begin{align*}
\mathcal{W}_{L+1}^{*}=-\mathcal{W}_{L+1},\ 
\mathcal{R}_{L+1}^{*}=\mathcal{R}_{L+1}. 
\end{align*}
Thus
\begin{align*}
	&\mathcal{U}_{L+1}^{*}\mathcal{U}_{L+1}=e^{-\mathcal{W}_{L+1}}\mathcal{U}_{L}^{*}\mathcal{U}_{L}e^{\mathcal{W}_{L+1}}=e^{-\mathcal{W}_{L+1}}e^{\mathcal{W}_{L+1}}=\bm1,\\
	&\mathcal{U}_{L+1}\mathcal{U}_{L+1}^{*}=e^{\mathcal{W}_{L+1}}\mathcal{U}_{L}\mathcal{U}_{L}^{*}e^{-\mathcal{W}_{L+1}}=e^{\mathcal{W}_{L+1}}e^{-\mathcal{W}_{L+1}}=\bm1,
\end{align*}
and
\begin{align*}
	\mathcal{M}_{L+1}^{*}&=\left(\mathcal{U}_{L+1}\mathcal{M}^{(L+1)}\mathcal{U}_{L+1}^{-1}+\mathcal{R}_{L+1}\right)^{*}\\
	&=\mathcal{U}_{L+1}\mathcal{M}^{(L+1)}\mathcal{U}_{L+1}^{-1}+\mathcal{R}_{L+1}=\mathcal{M}_{L+1}.
\end{align*}

This finishes the proof.
\end{proof}

\section{Proof of Theorem \ref{mthm}}
	In this section we prove the convergence of the  iterations (cf. Theorem \ref{main}) in the previous section  and thus finish the proof of Theorem \ref{mthm}.
	\begin{proof}[Proof of Theorem \ref{mthm}]
	Recalling Theorem \ref{main},  then for 
	\begin{align*}
	\theta_{0}^{-\delta}=	\|\mathcal{M}\|_{R,\alpha+3\delta}^{\frac{\delta}{\alpha-\alpha_{0}}}
\end{align*}
	and  $\theta_{0}>\eta_{0}$, we have 
	\begin{align*}
	\|\mathcal{M}\|_{R,\alpha+3\delta}=\theta_{0}^{\alpha_{0}-\alpha}<\eta_{0}^{\alpha_{0}-\alpha}=\ep_{0}.
\end{align*}
So applying Theorem \ref{main} yields  the existence of the sequence $\left(V_{l},\mathcal{M}_{l},\mathcal{W}_{l+1},\mathcal{R}_{l+1}\right)_{l=0}^{\infty}$. 
Thus $\hat{V}(z)-V(z)=\sum\limits_{j=0}^{\infty}\mathcal{M}_{j}(z,\bm0)$ exists in $\mathscr{H}_{\frac{R}{2}}$ and 
\begin{align*}
		\|\hat{V}-V\|_{\frac{R}{2},0}&=\|\hat{V}-V\|_{\frac{R}{2},\alpha_{0}}\le\sum_{j=0}^{\infty}2\theta_{j}^{\alpha_{0}-\alpha}\\
	&=\frac{2\theta_{0}^{\alpha_{0}-\alpha}}{1-\Theta^{\alpha_{0}-\alpha}}=K_{2}\|\mathcal{M}\|_{R,\alpha+3\delta},
\end{align*}
where $K_{2}=K_{2}(\alpha,\alpha_{0},\delta)=\frac{2}{1-\Theta^{\alpha_{0}-\alpha}}$.

Next, we prove the convergences  of $\mathcal{U}_{l+1}$ and $\mathcal{U}_{l+1}^{-1}$ with $\mathcal{U}_{l+1}=e^{\mathcal{W}_{l+1}}\cdots e^{\mathcal{W}_{1}}.$ Recalling \eqref{wl}, we can obtain
\begin{align*}
	\|\mathcal{W}_{l+1}\|_{R_{l+1},0}&\le\|\mathcal{W}_{l+1}\|_{R_{l+1},\alpha_{0}}\le\theta_{l}^{\alpha_{0}-\alpha+\tau+\delta}\le\theta_{l}^{-3\delta}
\end{align*}
and
\begin{align*}
		\|\mathcal{W}_{l+1}\|_{R_{l+1},\alpha-\tau-4\delta}\le\theta_{l}^{-3\delta}.
\end{align*}
Hence
\begin{align*}
		\sum_{l=0}^{\infty}\|\mathcal{W}_{l+1}\|_{R_{l+1},0}&\le\sum_{l=0}^{\infty}\theta_{l}^{-3\delta}=\frac{\theta_{0}^{-3\delta}}{1-\Theta^{-3\delta}}\le\frac{1}{1-\Theta^{-3\delta}}:=C<+\infty,\\
	\sum_{l=0}^{\infty}\|\mathcal{W}_{l+1}\|_{R_{l+1},\alpha-\tau-4\delta}&\le\sum_{l=0}^{\infty}\theta_{l}^{-3\delta}\le C.
\end{align*}
Then for $\forall\ \ep>0$, there is $N(\ep)\in\mathbb{N}$ so that for all $n\ge N$ and $p\in\mathbb{N}$,
\begin{align*}
	\sum_{l=n}^{n+p}\|\mathcal{W}_{l+1}\|_{R_{l+1},0}&< \frac{\ep}{4(1+C)K(\alpha)e^{K(\alpha)C}},\\
	\sum_{l=n}^{n+p}\|\mathcal{W}_{l+1}\|_{R_{l+1},\alpha-\tau-4\delta}&< \frac{\ep}{4(1+C)K(\alpha)e^{K(\alpha)C}}.
\end{align*}
Using  \eqref{Q} implies
\begin{align*}
	&\ \ \ \ \|e^{\mathcal{W}_{n+p+1}}\cdots e^{\mathcal{W}_{n+1}}-\bm1\|_{\frac{R}{2},\alpha-\tau-4\delta}\\
	&\le e^{K(\alpha-\tau-4\delta)\sum\limits_{l=n}^{n+p}\|\mathcal{W}_{l+1}\|_{R_{l+1},0}}	\sum_{l=n}^{n+p}\|\mathcal{W}_{l+1}\|_{R_{l+1},\alpha-\tau-4\delta}\\
	&\le e^{K(\alpha)\sum\limits_{l=n}^{n+p}\|\mathcal{W}_{l+1}\|_{R_{l+1},0}}	\sum_{l=n}^{n+p}\|\mathcal{W}_{l+1}\|_{R_{l+1},\alpha-\tau-4\delta},\\
	\|\mathcal{U}_{n}\|_{\frac{R}{2},0}&\le 1+\|e^{\mathcal{W}_{n}}\cdots e^{\mathcal{W}_{1}}-\bm1\|_{\frac{R}{2},0}\\
	&\le1+e^{K(\alpha)\sum\limits_{l=0}^{n-1}\|\mathcal{W}_{l+1}\|_{R_{l+1},0}}\sum_{l=0}^{n-1}\|\mathcal{W}_{l+1}\|_{R_{l+1},0}.
\end{align*}
As a result, we have
\begin{align*}
	&\ \ \ \|e^{\mathcal{W}_{n+p+1}}\cdots e^{\mathcal{W}_{n+1}}-\bm1\|_{\frac{R}{2},\alpha-\tau-4\delta}\|\mathcal{U}_{n}\|_{\frac{R}{2},0}\\
	&\le e^{K(\alpha)\sum\limits_{l=n}^{n+p}\|\mathcal{W}_{l+1}\|_{R_{l+1},0}}	\sum_{l=n}^{n+p}\|\mathcal{W}_{l+1}\|_{R_{l+1},\alpha-\tau-4\delta}\\
	&\ \ +e^{K(\alpha)\sum\limits_{l=0}^{n+p}\|\mathcal{W}_{l+1}\|_{R_{l+1},0}}\left(\sum_{l=0}^{n-1}\|\mathcal{W}_{l+1}\|_{R_{l+1},0}\right)\left(\sum_{l=n}^{n+p}\|\mathcal{W}_{l+1}\|_{R_{l+1},\alpha-\tau-4\delta}\right)\\
	&\le\frac{\ep}{2K(\alpha)}.
\end{align*}
Similarly, we get 
\begin{align*}
	\|e^{\mathcal{W}_{n+p+1}}\cdots e^{\mathcal{W}_{n+1}}-\bm1\|_{\frac{R}{2},0}\|\mathcal{U}_{n}\|_{\frac{R}{2},\alpha-\tau-4\delta}\le\frac{\ep}{2K(\alpha)}.
\end{align*}
From \eqref{t1},  we obtain
\begin{align*}
	\|\mathcal{U}_{n+p+1}-\mathcal{U}_{n}\|_{\frac{R}{2},\alpha-\tau-4\delta}&\le K(\alpha-\tau-4\delta)\|e^{\mathcal{W}_{n+p+1}}\cdots e^{\mathcal{W}_{n+1}}-\bm1\|_{\frac{R}{2},\alpha-\tau-4\delta}\|\mathcal{U}_{n}\|_{\frac{R}{2},0}\\
	&\ \ +K(\alpha-\tau-4\delta)\|e^{\mathcal{W}_{n+p+1}}\cdots e^{\mathcal{W}_{n+1}}-\bm1\|_{\frac{R}{2},0}\|\mathcal{U}_{n}\|_{\frac{R}{2},\alpha-\tau-4\delta}\\
	&\le K(\alpha)\left(\frac{\ep}{2K(\alpha)}+\frac{\ep}{2K(\alpha)}\right)=\ep,
\end{align*}
which implies the product $\mathcal{U}_{l+1}$  converges to some $\mathcal{U}\in\mathscr{U}_{\frac{R}{2},\alpha-\tau-4\delta}^{\bm\omega}$ as $\l \rightarrow \infty$. Similarly, one can show $\lim\limits_{l\rightarrow\infty}\|\mathcal{U}_{l+1}^{-1}-\mathcal{U}^{-1}\|_{\frac{R}{2},\alpha-\tau-4\delta}=0$. In addition, we have
\begin{align*}
	\|\mathcal{U}^{\pm 1}-\bm1\|_{\frac{R}{2},\alpha-\tau-4\delta}&\le\left(\sum_{l=0}^{\infty}\|\mathcal{W}_{l+1}\|_{R_{l+1},\alpha-\tau-4\delta}\right)e^{K(\alpha)\sum\limits_{l=0}^{\infty}\|\mathcal{W}_{l+1}\|_{R_{l+1},0}}\\
	&\le Ce^{K(\alpha)C}\theta_{0}^{-3\delta}=K_{1}\|\mathcal{M}\|_{R,\alpha+3\delta}^{\frac{3\delta}{\alpha-\alpha_{0}}},
\end{align*}
where $K_{1}=K_{1}(\alpha_{0},\alpha,\delta)=Ce^{K(\alpha)C}$.
	
	Next, considering $\sum\limits_{j=0}^{l}\mathcal{M}^{(j)}$, we have $\sum\limits_{j=0}^{l}\mathcal{M}^{(j)}=S_{\theta_{l}}\mathcal{M},$ 
	which implies
	\begin{align*}
		\|\mathcal{M}-\sum_{j=0}^{l}\mathcal{M}^{(j)}\|_{\frac{R}{2},\alpha-\tau-4\delta}&=\|(I-S_{\theta_{l}}\mathcal{M})\|_{\frac{R}{2},\alpha-\tau-4\delta}\le\theta_{l}^{-\tau-7\delta}\|\mathcal{M}\|_{{R},\alpha+3\delta}\\
		&\le\theta_{l}^{-\tau-7\delta}\rightarrow 0\  \text{(as $l\rightarrow\infty$)}.
	\end{align*}
Obviously, we have
\begin{align*}
	\|\mathcal{R}_{l+1}\|_{\frac{R}{2},\alpha-\tau-4\delta}&\le\|\mathcal{R}_{l+1}\|_{R_{l+1},\alpha-\tau-4\delta}\\
	&\le\theta_{l+1}^{-\tau-4\delta}\rightarrow0\  \text{(as $l\rightarrow\infty$)}.
\end{align*}
This finishes the proof of convergence of the iteration scheme. 
Moreover, from \eqref{vl},  we have $
	|\hat{V}|_{\frac{R}{2}}\ge\frac{1}{2}|V|_{R}.
$

The remaining is to show if both $V$ and $\mathcal{M}$ are self-adjoint, then $\mathcal{U}$ can be improved to become unitary and $\hat{V}^{*}=\hat{V}$. Suppose now 
\begin{align*}
	V^{*}=V,\ \mathcal{M}^{*}=\mathcal{M}.
\end{align*}
From  Theorem \ref{main}, we show  $\mathcal{U}_{l+1}$ is unitary and $V_{l}^{*}=V_{l}$ for each $l\ge0$. Therefore, 
\begin{align*}
	&\mathcal{U}^{*}\mathcal{U}=\lim\limits_{l\rightarrow\infty}\mathcal{U}_{l+1}^{*}\mathcal{U}_{l+1}=\bm1,\\
	&\mathcal{U}\mathcal{U}^{*}=\lim\limits_{l\rightarrow\infty}\mathcal{U}_{l+1}\mathcal{U}_{l+1}^{*}=\bm1,\\
	&\hat{V}^{*}=\lim\limits_{l\rightarrow\infty}V_{l}^{*}=\lim\limits_{l\rightarrow\infty}V_{l}=\hat{V},
\end{align*}
which implies $\mathcal{U}$ is unitary and $\hat{V}^{*}=\hat{V}$.

This completes the whole proof of Theorem \ref{mthm}.
\end{proof}

\section{Proof of Theorem \ref{pl}}
In this section we will prove the power-law localization (cf.  Theorem \ref{pl}) by using Theorem \ref{mthm}. 


We begin with a useful lemma. 
\begin{lem}
	For any $\mathcal{M},\mathcal{M}_{1},\mathcal{M}_{2}\in\mathscr{U}_{R,s}^{\bm\omega}$ and $z\in{D}_{R}$, we have
	\begin{align}
		\label{unitary}T_{\mathcal{M}^{*}}(z)&=\left(T_{\mathcal{M}}(\bar{z})\right)^{*},\\
		\label{tm12}T_{\mathcal{M}_{1}\mathcal{M}_{2}}(z)&=T_{\mathcal{M}_1}(z)T_{\mathcal{M}_2}(z).
	\end{align}
If in addition $s>q+\frac{d}{2}$, then
\begin{align}
	\label{bd}&\|T_{\mathcal{M}}(z)\|_q\le X(s,q)\|\mathcal{M}\|_{R,s},
\end{align} 
where
\begin{align}\label{csdq}
	X(s,q)=\sqrt{K(2q)(Y^2(s)+Y^2(s-q))}>0,
\end{align}
	 $K(s)$ is given by \eqref{cs}, $Y(s)=\sqrt{\sum\limits_{\bm n\in \Z^d}\langle \bm n\rangle^{-2s}}$ and the norm of $T_{\mathcal{M}}(z)$ denotes the standard operator norm on  $\ell_q^2(\Z^d)$.
	\end{lem}
\begin{proof}
Let $(\cdot,\cdot)$ denote the standard inner product on $\ell^{2}(\Z^d)$. 

First, for $\forall \ \psi$ and $\varphi\in\ell^{2}(\Z^d)$,
\begin{align*}
	\left(T_{\mathcal{M}^{*}}(z)\psi,\varphi\right)&=\sum_{\bm n\in\Z^d}\sum_{\bm l\in\Z^d}\mathcal{M}^{*}(z-\bm n\cdot\bm\omega,\bm l-\bm n)\psi(\bm l)\overline{\varphi(\bm n)}\\
	&=\sum_{\bm n\in\Z^d}\sum_{\bm l\in\Z^d}\psi(\bm l)\overline{\mathcal{M}(\bar{z}-\bm l\cdot\bm\omega,\bm n-\bm l)\varphi(\bm n)}\\
	&=\left(\psi,T_{\mathcal{M}}(\bar{z})\varphi\right)=\left(\left(T_{\mathcal{M}}(\bar{z})\right)^{*}\psi,\varphi\right),
\end{align*}
which shows \eqref{unitary}. 

Next, for $\forall\  \psi\in\ell^{2}(\Z^d)$ and $\forall\ \bm n\in\Z^d$, we have 
\begin{align*}
	\left(T_{\mathcal{M}_{1}\mathcal{M}_{2}}(z)\psi\right)(\bm n)&=\sum_{\bm l\in\Z^d}\sum_{\bm k\in\Z^d}\mathcal{M}_{1}(z-\bm n\cdot\bm\omega,\bm k)\mathcal{M}_{2}(z-(\bm n+\bm k)\cdot\bm\omega,\bm l-(\bm n+\bm k))\psi(\bm l)\\
	&=\sum_{\bm k\in\Z^d}\mathcal{M}_{1}(z-\bm n\cdot\bm\omega,\bm k)\left(T_{\mathcal{M}_{2}}\psi\right)(\bm n+\bm k)\\
	&=\sum_{\bm k\in\Z^d}\mathcal{M}_{1}(z-\bm n\cdot\bm\omega,\bm k-\bm n)\left(T_{\mathcal{M}_{2}}\psi\right)(\bm k)\\
	&=\left(T_{\mathcal{M}_{1}}\left(T_{\mathcal{M}_{2}}\psi\right)\right)(\bm n)=\left(\left(T_{\mathcal{M}_{1}}T_{\mathcal{M}_{2}}\right)\psi\right)(\bm n), 
\end{align*}
which implies \eqref{tm12}. 

	If in addition $s>q+\frac{d}{2}$ for $\forall \ \psi\in\ell_{q}(\Z^d)$, we get
\begin{align*}
	\|T_{\mathcal{M}}(z)\psi\|^{2}_{q}&=\sum_{\bm n\in\Z^d}\left|\sum_{\bm l\in\Z^{d}}\mathcal{M}(z-\bm n\cdot\bm \omega,\bm l-\bm n)\psi(\bm l)\right|^{2}\langle\bm n\rangle^{2q}\\
	&\le\sum_{\bm n\in\Z^d}\left(\sum_{\bm l\in\Z^d}|\mathcal{M}(z-\bm n\cdot\bm \omega,\bm l-\bm n)||\psi(\bm l)|\right)^{2}\langle\bm n\rangle^{2q}.
\end{align*}
Recalling \eqref{norm}, we obtain
\begin{align*}
	|\mathcal{M}(z-\bm n\cdot\bm\omega,\bm l-\bm n)|
	&\le\langle\bm l-\bm n\rangle^{-s}\|\mathcal{M}\|_{R,s}
\end{align*}
By Cauchy-Schwarz inequality, we have
\begin{align*}
	\left(\sum_{\bm l\in\Z^d}|\mathcal{M}(z-\bm n\cdot\bm \omega,\bm l-\bm n)||\psi(\bm l)|\right)^{2} &\le\left(\sum_{\bm l\in\Z^d}|\mathcal{M}(z-\bm n\cdot\bm \omega,\bm l-\bm n)|\langle\bm l-\bm n\rangle^{s}\right)\\
	&\ \ \times\left(\sum_{\bm l\in\Z^d}|\mathcal{M}(z-\bm n\cdot\bm \omega,\bm l-\bm n)|\langle\bm l-\bm n\rangle^{-s}|\psi(\bm l)|^{2}\right)\\
	&\le\|\mathcal{M}\|_{R,s}^{2}\left(\sum_{\bm l\in\Z^d}\langle\bm l-\bm n\rangle^{-2s}|\psi(\bm l)|^{2}\right),
\end{align*}
which implies
\begin{align*}
	\|T_{\mathcal{M}}(z)\psi\|^{2}_{q}\le\|\mathcal{M}\|_{R,s}^{2}\sum_{\bm n\in\Z^d}\left(\sum_{\bm l\in\Z^d}\langle\bm l-\bm n\rangle^{-2s}|\psi(\bm l)|^{2}\langle\bm n\rangle^{2q}\right).
\end{align*}
According to \eqref{xy}, we have
\begin{align*}
&\ \ \ \ \sum_{\bm n\in\Z^d}\left(\sum_{\bm l\in\Z^d}\langle\bm l-\bm n\rangle^{-2s}|\psi(\bm l)|^{2}\langle\bm n\rangle^{2q}\right)\\
&\le K(2q)\sum_{\bm n\in\Z^d}\left(\sum_{\bm l\in\Z^d}\langle\bm l-\bm n\rangle^{-2s}|\psi(\bm l)|^{2}\left(\langle\bm l\rangle^{2q}+\langle\bm l-\bm n\rangle^{2q}\right)\right)\\
&=K(2q)\left(\sum_{\bm n\in\Z^d}\langle\bm n\rangle^{-2s}\sum_{\bm l\in\Z^d}|\psi(\bm l)|^{2}\langle\bm l\rangle^{2q}+\sum_{\bm n\in\Z^d}\langle\bm n\rangle^{-(2s-2q)}\sum_{\bm l\in\Z^d}|\psi(\bm l)|^{2}\right)\\
&=K(2q)(Y^2(s)\|\psi\|_{q}^{2}+Y^2(s-q)\|\psi\|_{0}^{2})\\
&\le K(2q)(Y^2(s)+Y^2(s-q))\|\psi\|_{q}^{2}.
\end{align*}
Therefore,
\begin{align*}
	\|T_{\mathcal{M}}(z)\psi\|^{2}_{q}\le\|\mathcal{M}\|_{R,s}^{2}K(2q)(Y^2(s)+Y^2(s-q))\|\psi\|_{q}^{2},
\end{align*}
that is 
\begin{align*}
	\|T_{\mathcal{M}}(z)\|_q\le X(s,q)\|\mathcal{M}\|_{R,s}.
\end{align*}
\end{proof}

We can now prove Theorem \ref{pl}.

\begin{proof}[Proof of Theorem \ref{pl}]
	We apply Theorem \ref{mthm} with
	\begin{align}\label{as}
		\alpha=s-3\delta.
	\end{align}
Since \eqref{as}, we have
\begin{align*}
	\mathcal{M}&\in\mathscr{U}_{R,s}^{\bm\omega}=\mathscr{U}_{R,\alpha+3\delta}^{\bm\omega},\\
	\alpha&=s-3\delta>\alpha_{0}+\tau+4\delta+\frac{d}{2}>\alpha_{0}+\tau+4\delta.
\end{align*}
Hence applying Theorem \ref{mthm} implies that if $ \|\mathcal{M} \|_{R,s}\ll1$, there are  $\mathcal{U}\in\mathscr{U}_{\frac{R}{2},s-\tau-7\delta}^{\bm\omega}$ and $\hat{V}\in\mathscr{P}_{\frac{R}{2}}$ so that
\begin{align*}
	 \mathcal{U}(V+\mathcal{M})\mathcal{U}^{-1}&=\hat{V},\\
	\|\mathcal{U}^{\pm1}-\bm 1\|_{\frac{R}{2},s-\tau-7\delta}&\le K_{1}\|\mathcal{M}\|_{R,s}^{\frac{3\delta}{s-\alpha_{0}-3\delta}}.
\end{align*}
Letting $U_{z}=T_{\mathcal{U}}(z)$,  according to \eqref{bd} and $s-\tau-7\delta>\frac{d}{2}$, we obtain
\begin{align*}
	\|U_{z}^{\pm1}\|_{0}&\le 1+\|U_{z}^{\pm1}-I\|_{0}\le1+X(s-\tau-7\delta,0) \|\mathcal{U}^{\pm1}-\bm1\|_{\frac{R}{2},s-\tau-7\delta}\\
	&\le1+X(s-\tau-7\delta,0)K_{1}\|\mathcal{M}\|_{R,s}^{\frac{3\delta}{s-\alpha_{0}-3\delta}},
\end{align*} 
where $X(s-\tau-7\delta,0)$ and $K_{1}$ are given by \eqref{csdq} and \eqref{K1} respectively. Therefore $U_{z}$ is a bounded invertible operator. By \eqref{tm12},  we get for $z\in\mathcal{Z}_{{R}/{2}}$ (cf. \eqref{ZR}), 
\begin{align*}
	U_{z}H_{z}U_{z}^{-1}=T_{\hat{V}}(z).
\end{align*}

Note that $T_{\hat{V}}(z)$ is a diagonal operator for $z\in\mathcal{Z}_{R/2}$. Then the standard basis $\{\delta_{\bm n}\}_{\bm n\in\Z^d}$ of $\ell^{2}(\Z^d)$ is a complete set of eigenfunctions of $T_{\hat{V}}(z)$ with eigenvalues $\{\hat{V}(z-\bm n\cdot\bm\omega)\}_{\bm n\in\Z^d}$ for  $z\in\mathcal{Z}_{R/2}$. 
Letting $\varphi_{\bm n}=U_{z}^{-1}\delta_{\bm n}$, then
\begin{align*}
	H_{z}\varphi_{\bm n}
	=\hat{V}(z-\bm n\cdot\bm\omega)\varphi_{\bm n},
\end{align*}
which combined with the boundedness of $U_{z}^{-1}$ will imply  $\{\varphi_{\bm n}\}_{\bm n\in\Z^d}$ is a \textit{complete set of eigenfunctions} of $H_{z}$ for $z\in\mathcal{Z}_{R/2}$.  In fact, we have
\begin{align*}
	\varphi_{\bm n}(\bm i)=\sum_{\bm l\in\Z^d}\mathcal{U}^{-1}(z-\bm i\cdot\bm\omega,\bm l-\bm i)\delta_{\bm n}(\bm l)=\mathcal{U}^{-1}(z-\bm i\cdot\bm\omega,\bm n-\bm i).
\end{align*}
Since $\mathcal{U}^{-1}\in\mathscr{U}_{\frac{R}{2},s-\tau-7\delta}^{\bm\omega}$, we obtain
\begin{align*}
	|\varphi_{\bm n}(\bm i)|&\le\|\mathcal{U}^{-1}\|_{\frac{R}{2},s-\tau-7\delta}\langle\bm n-\bm i\rangle^{-s+\tau+7\delta}\\
	&\le(1+\|\mathcal{U}^{-1}-\bm 1\|_{\frac{R}{2},s-\tau-7\delta})\langle\bm n-\bm i\rangle^{-s+\tau+7\delta}\\
	&\le(1+K_{1}\|\mathcal{M}\|_{R,s}^{\frac{3\delta}{s-\alpha_{0}-3\delta}})\langle\bm n-\bm i\rangle^{-s+\tau+7\delta}\\
	&\le2\langle\bm n-\bm i\rangle^{-s+\tau+7\delta}
\end{align*}
provided $K_{1}\|\mathcal{M}\|_{R,s}^{\frac{3\delta}{s-\alpha_{0}-3\delta}}\le1$. This shows particularly $\varphi_{\bm n}\in\ell^2(\Z^d).$
We then prove the \textit{completeness}.   Suppose for all $\bm n\in\Z^d$, $(\psi,\varphi_{\bm n})=0$. It suffices to show $\psi=0$.  For $\forall\ \bm n\in\Z^d$, we have
\begin{align*}
	0=(\psi,\varphi_{\bm n})=\psi(\bm n)+(\psi,(U_{z}^{-1}-I)\delta_{\bm n}),
\end{align*}
which combined with \eqref{unitary} implies
\begin{align*}
	\|\psi\|^{2}_{0}=\sum_{\bm n\in\Z^d}|\psi(\bm n)|^{2}&=\sum_{\bm n\in\Z^d}|(\psi,(U_{z}^{-1}-I)\delta_{\bm n})|^{2}\\
	&=\sum_{\bm n\in\Z^d}|((U_{z}^{-1}-I)^{*}\psi,\delta_{\bm n})|^{2}\\
	&=\sum_{\bm n\in\Z^d}\left|\left(T_{(\mathcal{U}^{-1}-\bm1)^{*}}(\bar{z})\psi\right)(\bm n)\right|^{2}\\
	&=\|T_{(\mathcal{U}^{-1}-\bm1)^{*}}(\bar{z})\psi\|^{2}_{0}.
\end{align*}
By using \eqref{bd} and $ \|\mathcal{M} \|_{R',s'}=\|\mathcal{M}^{*} \|_{R',s'}$, we obtain
\begin{align*}
	\|T_{(\mathcal{U}^{-1}-\bm1)^{*}}(\bar{z})\psi\|^{2}_{0}&\le X^{2}(s-\tau-7\delta,0)\|(\mathcal{U}^{-1}-\bm 1)^{*}\|_{\frac{R}{2},s-\tau-7\delta}^{2}\|\psi\|^{2}_{0}\\
	&=X^{2}(s-\tau-7\delta,0)\|(\mathcal{U}^{-1}-\bm 1)\|_{\frac{R}{2},s-\tau-7\delta}^{2}\|\psi\|^{2}_{0}\\
	&\le K_{1}^{2}X^{2}(s-\tau-7\delta,0)\|\mathcal{M}\|_{R,s}^{\frac{6\delta}{s-\alpha_{0}-3\delta}}\|\psi\|^{2}_{0}.
\end{align*}
Thus, we have
\begin{align*}
\|\psi\|^{2}_{0}\le K_{1}^{2}X^{2}(s-\tau-7\delta,0)\|\mathcal{M}\|_{R,s}^{\frac{6\delta}{s-\alpha_{0}-3\delta}}\|\psi\|^{2}_{0}\le \frac{1}{2}\|\psi\|^{2}_{0}
\end{align*}
provided $K_{1}^{2}X^{2}(s-\tau-7\delta,0)\|\mathcal{M}\|_{R,s}^{\frac{6\delta}{s-\alpha_{0}-3\delta}}\le\frac{1}{2}$. This shows $\psi=0$ and  thus the  \textit{completeness} of $\{\varphi_{\bm n}\}_{\bm n\in\Z^d}$.

Finally, if $\mathcal{M}$ and $V$ are self-adjoint,  then $\mathcal{U}$ is unitary and $\hat{V}^{*}=\hat{V}$. Therefore, by \eqref{tm12} and \eqref{unitary}, we have for $x\in\mathcal{Z}_0$
\begin{align*}
	(H_{x})^{*}&={T_{\mathcal{M}^{*}}(\bar{x})+T_{V^{*}}(\bar{x})=T_{\mathcal{M}}(x)+T_{V}(x)}=H_{x},\\
	(U_{x})^{*}U_{x}&=T_{\mathcal{U}^{*}}(\bar{x})T_{\mathcal{U}}(x)=T_{\mathcal{U}^{*}\mathcal{U}}(x)=I,\\
	U_{x}(U_{x})^{*}&=T_{\mathcal{U}}(x)T_{\mathcal{U}^{*}}(\bar{x})=T_{\mathcal{U}\mathcal{U}^{*}}(x)=I,
\end{align*}
which implies $H_{x}$ is a self-adjoint operator and $U_{x}$ is a unitary operator. Hence the spectrum of $H_{x}$ is equal to that of $T_{\hat{V}}(x)$. At this stage, we recall a result proven in \cite{BLS83}.
\begin{lem}\label{bls}
	If $g\in\mathscr{P}_{R}$, $g^{*}=g$, then there is a unique $x\in[0,1)$ such that the real poles of $g$ are $\{x+n:n\in\Z\}$. Moreover, $g$ is strictly monotone in each interval $(x+n,x+n+1)$ and $\mathbb{R}=\{g(x):x\in\mathbb{R}\}$.
\end{lem}
Since $\{\bm n\cdot\bm\omega \mod 1:\  \bm n\in\Z^d\}$ is dense in $[0,1]$, we have
\begin{align*}
	\sigma(T_{\hat{V}}(x))=\mathbb{R},
 \end{align*}
Consequently, we obtain $\sigma(H_{x})=\mathbb{R}$ for $x\in\mathcal{Z}_0$.

This proves Theorem \ref{pl}.
\end{proof}

\section{Proof of Theorem \ref{dl}}
In this section we prove the uniform  dynamical localization by applying Theorem \ref{mthm}. 
\begin{proof}[Proof of Theorem \ref{dl}]
From Theorem  \ref{mthm},  there  exist  $\mathcal{U}\in\mathscr{U}_{\frac{R}{2},s-\tau-7\delta}^{\bm\omega}$  and $\hat{V}\in\mathscr{P}_{\frac{R}{2}}$ so that
\begin{align*}
	\mathcal{U}(V+\mathcal{M})\mathcal{U}^{-1}&=\hat{V},\\
	\|V-\hat{V}  \|_{\frac{R}{2},0}&\le K_{2}\|\mathcal{M}\|_{R,\alpha+3\delta},
\end{align*}
Moreover,  for $\forall\  x\in\mathcal{Z}_0$, 
\begin{align*}
	U_{x}H_{x}U_{x}^{-1}=T_{\hat{V}}(x),
\end{align*}
where $\hat{V}(x)$ is real valued. 
Thus  for $\forall \ \psi\in \ell^2_{q}(\Z^d)$ and $x\in\mathcal{Z}_0$, we have
\begin{align*}
	\|e^{-\sqrt{-1}tT_{\hat{V}}(x)}\psi\|_{q}^{2}&=\sum_{\bm n\in\Z^d}|e^{-\sqrt{-1}t\hat{V}(x-\bm n\cdot\bm\omega)}\psi(\bm n)|^{2}\langle\bm n\rangle^{2q}=\|\psi\|_q^2. 
\end{align*}
Recalling  \eqref{bd} and $s-\tau-7\delta>q+\frac{d}{2}$, $U_{x}$ is a bounded invertible operator on $\ell_{q}^2(\Z^d)$.  So for $\forall\  \psi\in \ell_{q}^2(\Z^d)$ and $x\in\mathcal{Z}_0$, we have 
\begin{align*}
	\|e^{-\sqrt{-1}tH_x}\psi\|_{q}^{2}&=\|U_{x}^{-1}e^{-\sqrt{-1}tT_{\hat{V}}(x)}U_{x}\psi\|_{q}^{2}\le\|U_{x}^{-1}\|_{q}^{2}\|e^{-\sqrt{-1}tT_{\hat{V}}(x)}U_{x}\psi\|_{q}^{2}\\
	&=\|U_{x}^{-1}\|_{q}^{2}\sum_{\bm n\in\Z^d}|e^{-\sqrt{-1}t\hat{V}(x-\bm n\cdot\bm\omega)}(U_{x}\psi)(\bm n)|^{2}\langle\bm n\rangle^{2q}\\
	&\le\|U_{x}^{-1}\|_{q}^{2}\|U_{x}\psi\|_{q}^{2}\le\|U_{x}^{-1}\|_{q}^{2}\|U_{x}\|_{q}^{2}\|\psi\|_{q}^{2}\\
	&\le X^4(s-\tau-7\delta,q)\|\mathcal{U}\|_{\frac{R}{2},s-\tau-7\delta}^{4}\|\psi\|^2_{q}<\infty.
\end{align*}

This proves Theorem \ref{dl}.

\end{proof}

\section{Proof of Theorem \ref{idsc}}
In this section we prove Lipschitz continuity of the IDS by using Theorem \ref{mthm}.
We first prove the invariance of IDS under unitary transformation that is also nearly identical. 

For any self-adjoint operator $H$ defined on $\ell^2(\Z^d)$ and $E\in\R$,  let
$$\kappa_{H}(E)=\lim_{L\to\infty}\frac{1}{(2L+1)^d}\text{tr}(\chi_{L}\mathbb{P}_{(-\infty,E]}(H)).$$
We have
\begin{lem}\label{idsu}
 	Let  $H$ be a self-adjoint operator on $\ell^2(\Z^d)$ and let  $UHU^{-1}=D$,  where $U$  is unitary and $D$ is diagonal. Assume  that the matrix elements  $u_{\bm i\bm j}=(\delta_{\bm i}, U\delta_{\bm j})$ of $U$  satisfy
 	\begin{align*}
 		|u_{\bm i \bm j}-\delta_{\bm i\bm j}|\le c_{1}\langle \bm i-\bm j\rangle^{-r},
 	\end{align*}
 where $r>d$ and $c_1>0$. Then 
 \begin{align*}
 	\kappa_{H}(E)=\kappa_{D}(E) 
 \end{align*}
if one of them exists.
 \end{lem}
\begin{proof}
	By $UHU^{-1}=D$, we have
	\begin{align*}
		\text{tr}(\chi_{L}\mathbb{P}_{(-\infty,E]}(H))&=\text{tr}(\chi_{L}U^{-1}\mathbb{P}_{(-\infty,E]}(D)U)\\
		&=\text{tr}(U\chi_{L}U^{-1}\mathbb{P}_{(-\infty,E]}(D))\\
		&=\text{tr}(\chi_{L}\mathbb{P}_{(-\infty,E]}(D)+(U\chi_{L}-\chi_{L}U)U^{-1}\mathbb{P}_{(-\infty,E]}(D)).
	\end{align*}
	Direct computations yield 
	\begin{align*}
		(U\chi_{L}-\chi_{L}U)_{\bm i\bm j}=\left\{
		\begin{array}{cc}
			0 & \text{if}\ |\bm i|>L\ \text{and}\ |\bm j|>L,\\
			0 & \text{if}\ |\bm i|\le L\ \text{and}\ |\bm j|\le L,\\
			-u_{\bm i\bm j} & \text{if}\ |\bm i|\le L\ \text{and}\ |\bm j|>L,\\
			u_{\bm i\bm j} & \text{if}\ |\bm i|>L\ \text{and}\ |\bm j|\le L.
		\end{array}
		\right.
	\end{align*}
	Denote  $(U^{-1}\mathbb{P}_{(-\infty,E]}(D))_{\bm i\bm j}=c_{\bm i\bm j}$.  Then $\sup_{\bm i,\bm j\in\Z^d}|c_{\bm i\bm  j}|\leq c_2$ for some $c_2>0$.  Note also that 
	\begin{align*}
		|((U\chi_{L}-\chi_{L}U)U^{-1}\mathbb{P}_{(-\infty,E]}(D))_{\bm i\bm j}|\le \left\{
		\begin{array}{cc}
			\sum\limits_{|\bm k|>L}|u_{\bm i\bm k}c_{\bm k\bm j}| & \text{if}\ |\bm i|\le L,\\
			\sum\limits_{|\bm k|\le L}|u_{\bm i\bm k}c_{\bm k\bm j}| & \text{if}\ |\bm i|> L.
		\end{array}
		\right.
	\end{align*}
	As a result, the decay property of  $u_{\bm i\bm j}$ allows us to obtain 
	\begin{align*}
		|\text{tr}((U\chi_{L}-\chi_{L}U)U^{-1}\mathbb{P}_{(-\infty,E]}(D))|&\le c_2\sum_{|\bm i|\le L}\sum_{|\bm k|>L}|u_{\bm i\bm k}| +c_2\sum_{|\bm i|> L}\sum_{|\bm k|\le L}|u_{\bm i\bm k}|\\
		&\le2c_{1}c_2\sum_{|\bm i|\le L}\sum_{|\bm k|>L}\langle\bm i-\bm k\rangle^{-r}.
	\end{align*}
	To control the above sum, we have the following two cases.
	
	\textbf{Case 1.} $|\bm i|< L-\sqrt{L}$. In this case  $\langle \bm i-\bm k\rangle\ge \sqrt{L}+1$ for $|\bm k|>L$.  
	Note that for any $\bm i\in\Z^d$ and $t\in \N$,
	\begin{align*}
		\#\{\bm k\in\Z^d:\ \langle\bm i-\bm k\rangle=t,|\bm k|>L\}
		\le(2t+1)^{d}-(2t-1)^{d}\le 2d(3t)^{d-1}.
	\end{align*}
	Hence, for a fixed $\bm i\in\Z^d$ with $|\bm i|<L-\sqrt{L}$, we get
	\begin{align*}
		\sum_{|\bm k|>L}\langle\bm i-\bm k\rangle^{-r}&={\sum_{t\ge\sqrt{L}+1}}\sum_{\langle\bm i-\bm k\rangle=t,|\bm k|>L}\langle\bm i-\bm k\rangle^{-r}\\
		&={\sum_{t\ge\sqrt{L}+1}}\frac{\#\{\bm k\in\Z^d:\ \langle\bm i-\bm k\rangle=t,|\bm k|>L\}}{t^{r}}\\
		&\le{\sum_{t\ge\sqrt{L}+1}}\frac{2d3^{d-1}}{t^{r-d+1}}.
	\end{align*} 
	Since $r>d$ and fixing   any  $\eta\in(0,r-d)$,  we have 
	\begin{align*}
		\sum_{|\bm i|\le L-\sqrt{L}}\sum_{|\bm k|>L}\langle\bm i-\bm k\rangle^{-r}&\le\#\{\bm i\in\Z^d:\ |\bm i|\le L\}{\sum_{t\ge\sqrt{L}+1}}\frac{2d3^{d-1}}{t^{t-d+1}}\\
		&\le\frac{(2L+1)^d}{\sqrt{L}^{\eta}}{\sum_{t\ge\sqrt{L}+1}}\frac{2d3^{d-1}}{t^{r-d-\eta+1}}\\
		&\le\frac{(2L+1)^d}{\sqrt{L}^\eta}\sum_{t=1}^{\infty}\frac{2d3^{d-1}}{t^{r-d-\eta+1}}=O(L^{d-\eta/2}).
	\end{align*} 
	
	\textbf{Case 2.} $L-\sqrt{L}\le|\bm i|\le L$.  In this case we have
	\begin{align*}
		\sum_{|\bm k|>L}\langle\bm i-\bm k\rangle^{-r}&=\sum_{t=1}^{\infty}\sum_{\langle\bm i-\bm k\rangle=t,|\bm k|>L}\langle\bm i-\bm k\rangle^{-r}\\
		&=\sum_{t=1}^{\infty}\frac{\#\{\bm k\in\Z^d:\ \langle\bm i-\bm k\rangle=t,|\bm k|>L\}}{t^{r}}\\
		&\le\sum_{t=1}^{\infty}\frac{2d3^{d-1}}{t^{r-d+1}}.
	\end{align*}
	Thus, we obtain 
	\begin{align*}
		\sum_{L-\sqrt{L}\le|\bm i|\le L}\sum_{|\bm k|>L}\langle\bm i-\bm k\rangle^{-r}&\le\#\{\bm i\in\Z^d:\ L-\sqrt{L}\le|\bm i|\le L\}\sum_{t=1}^{\infty}\frac{2d3^{d-1}}{t^{t-d+1}}=O(L^{d-1/2}).
	\end{align*}
	So combining \textbf{Case 1} and \textbf{Case 2} implies
	\begin{align*}
		\lim\limits_{L\rightarrow\infty}\frac{1}{(2L+1)^{d}}\text{tr}((U\chi_{L}-\chi_{L}U)U^{-1}\mathbb{P}_{(-\infty,E]}(D))=0,
	\end{align*}
	which concludes the proof of this  lemma. 

\end{proof}

 We are ready to prove Theorem \ref{idsc}. 
\begin{proof}[Proof of Theorem \ref{idsc}]
From Theorem \ref{mthm},  there exist  $\mathcal{U}\in\mathscr{U}_{\frac{R}{2},s-\tau-7\delta}^{\bm\omega}$ and $\hat{V}\in\mathscr{P}_{\frac{R}{2}}$ so that
	\begin{align*}
		\mathcal{U}(V+\mathcal{M})\mathcal{U}^{-1}&=\hat{V},\\
U_{x}H_{x}U_{x}^{-1}&=T_{\hat{V}}(x),
\end{align*}
where  $U_{x}$  is unitary  and $T_{\hat{V}}(x)$ is diagonal for $x\in\mathcal{Z}_0$. Recalling \eqref{K1}, we obtain
\begin{align*}
	|(U_{x})_{\bm i\bm j}-\delta_{\bm i\bm j}|&=|(\mathcal{U}-\bm1)(x-\bm i\cdot\bm\omega,\bm j-\bm i)|\\
	&\le\|\mathcal{U}-\bm 1\|_{\frac{R}{2},s-\tau-7\delta}\langle\bm i-\bm j\rangle^{-s+\tau+7\delta}\\
		&\le K_{1}\|\mathcal{M}\|_{R,s}^{\frac{3\delta}{s-\alpha_{0}-3\delta}}\langle\bm i-\bm j\rangle^{-(s-\tau-7\delta)},
\end{align*}
where $s-\tau-7\delta>d$. Denote by $\kappa_0(E)$ the IDS of $T_{\hat{V}}(x)$. From Lemma \ref{idsu} and $\bm\omega\in{\rm DC}_{\tau,\gamma}$, 
 it follows  that 
\begin{align}\label{kap0}
	\kappa(E)=\kappa_{0}(E)={\rm mes}\{\theta\in \T:\ \hat{V}(\theta)\le E\},
\end{align}
where ${\rm mes}(\cdot)$ denotes the Lebesgue measure.  
 Recalling Lemma \ref{bls} and  $\hat{V}^{*}=\hat{V}$,   we assume  $\hat{V}$ is non-decreasing  in $(0,1)$ without loss of generality.    Then  $\hat{V}$ has  an  inverse function defined on $\R$ which is denoted by $\hat{V}^{-1}$. 
 Since $\hat{V}\in\mathscr{P}_{\frac{R}{2}}$ and \eqref{vr2},  
we have
\begin{align*}
	\inf_{x\in(0,1)}\frac{d\hat{V}(x)}{dx}\ge|\hat{V}|_{\frac{R}{2}}\ge\frac{1}{2}|V|_{R}>0,
\end{align*}
which combined with \eqref{kap0} implies for $E_1<E_2$, 
\begin{align*}
	\left|\kappa(E_{2})-\kappa(E_{1})\right|&={\rm mes}\{\theta\in \T:\ E_{1}<\hat{V}(\theta)\le E_{2}\}\\
	&={\rm mes}\{\theta\in \T:\ \hat{V}^{-1}(E_{1})<\theta\le \hat{V}^{-1}(E_{2})\}\\
	&=\hat{V}^{-1}(E_{2})-\hat{V}^{-1}(E_{1})\\
	&\le\frac{2}{|V|_{R}}(E_{2}-E_{1}).
\end{align*}

This proves Theorem \ref{idsc}. 
\end{proof}

\section*{Acknowledgments}
  This work  is partially supported by 
   the NSF of China (No. 12271380).

\appendix{}

\section{}
\begin{proof}[Proof of Lemma \ref{ts}]
	Recalling \eqref{norm} and \eqref{xy}, we have for any $\mathcal{M}_{1}$ and $\mathcal{M}_{2}\in\mathscr{U}_{R,s}^{\bm\omega}$, 
	\begin{align*}
		\|\mathcal{M}_{1}\mathcal{M}_{2}\|_{R,s}\le&\underset{z\in{D}_{R}}{\sup}\sum_{\bm l\in\Z^d}\sum_{\bm n\in\Z^d}|\mathcal{M}_{1}(z,\bm l)||\mathcal{M}_{2}(z-\bm l\cdot\bm\omega,\bm n-\bm l)|\left(\langle\bm l\rangle+\langle\bm n-\bm l\rangle\right)^{s}\\
		\le& K(s)\underset{z\in{D}_{R}}{\sup}\sum_{\bm l\in\Z^d}\sum_{\bm n\in\Z^d}|\mathcal{M}_{1}(z,\bm l)||\mathcal{M}_{2}(z-\bm l\cdot\bm\omega,\bm n-\bm l)|\left(\langle\bm l\rangle^{s}+\langle\bm n-\bm l\rangle^{s}\right)\\
		\le&K(s)\left(\underset{z\in{D}_{R}}{\sup}\sum_{\bm l\in\Z^d}|\mathcal{M}_{1}(z,\bm l)|\langle\bm l\rangle^{s}\|\mathcal{M}_{2}\|_{R,0}+\underset{z\in{D}_{R}}{\sup}\sum_{\bm l\in\Z^d}|\mathcal{M}_{1}(z,\bm l)|\|\mathcal{M}_{2}\|_{R,s}\right)\\
		\le&K(s)\left(\|\mathcal{M}_{1}\|_{R,s}\|\mathcal{M}_{2}\|_{R,0}+\|\mathcal{M}_{1}\|_{R,0}\|\mathcal{M}_{2}\|_{R,s}\right),
	\end{align*}
	which implies  Lemma \ref{ts}. 
\end{proof}

\begin{proof}[Proof of Lemma \ref{ew}]
	
		Recalling \eqref{t1}, we have 
	\begin{align*}
	\|\mathcal{N}_{1}^{l_{1}}\cdots\mathcal{N}_{k}^{l_{k}}\|_{R',s}\le(K(s))^{l_{1}+\cdots+l_{k}-1}\sum_{m=1}^{k}l_{m}\left(\prod_{j\ne m}\|\mathcal{N}_{j}\|_{R',0}^{l_{j}}\right)\|\mathcal{N}_{m}\|_{R',0}^{l_{m}-1}\|\mathcal{N}_{m}\|_{R',s}.
	\end{align*}
	which implies
	\begin{align*}
	&\ \ \ \ \left\|e^{\mathcal{N}_{1}}\cdots e^{\mathcal{N}_{k}}-\bm1\right\|_{R',s}\\
	=&\left\|\sum_{l=1}^{\infty}\sum_{l_{1}+\cdots+l_{k}=l}\frac{\mathcal{N}_{1}^{l_{1}}\cdots\mathcal{N}_{k}^{l_{k}}}{l_{1}!\cdots l_{k}!}\right\|_{R',s}\\
	\le&\sum_{l=1}^{\infty}(K(s))^{l-1}\sum_{l_{1}+\cdots+l_{k}=l}\sum_{m=1}^{k}\frac{l_{m}\left(\prod_{j\ne m}\|\mathcal{N}_{j}\|_{R',0}^{l_{j}}\right)\|\mathcal{N}_{m}\|_{R',0}^{l_{m}-1}\|\mathcal{N}_{m}\|_{R',s}}{l_{1}!\cdots l_{k}!}\\
	=&\sum_{m=1}^{k}\sum_{l=1}^{\infty}\frac{(K(s))^{l-1}}{(l-1)!}\left(\sum_{j=1}^{k}\|\mathcal{N}_{j}\|_{R',0}\right)^{l-1}\|\mathcal{N}_{m}\|_{R',s}\\
	\leq&e^{K(s)\left(\sum\limits_{m=1}^{k}\|\mathcal{N}_{m}\|_{R',0}\right)}\left(\sum\limits_{m=1}^{k}\|\mathcal{N}_{m}\|_{R',s}\right).
	\end{align*}
\end{proof}
\bibliographystyle{alpha}

 \end{document}